\documentclass[a4paper,10pt]{amsart}
\usepackage[utf8]{inputenc}
\usepackage{amsmath}
\usepackage{amsfonts}
\usepackage{amsthm}
\usepackage{amssymb}
\usepackage{xcolor}

\newcommand{\p}[1]{\mathbb{P}}

\newcommand{\fa}[1]{\forall}

\newcommand{\bP}{\mathbb{P}}

\newcommand{\ti}{\times}

\newcommand{\N}{\mathbb{N}}
\newcommand{\R}{\mathbb{R}}

\newcommand{\hd}[2]{d_{H}( #1, #2 )}

\theoremstyle{definition}
\newtheorem{defi}{Definition}[section]

\theoremstyle{remark}

\theoremstyle{plain}
\newtheorem{lemma}[defi]{Lemma}
\newtheorem{coro}[defi]{Corollary}

\newtheorem{thm}[defi]{Theorem}

\allowdisplaybreaks

\newcommand{\sgn}{\operatorname{sign}}

\newcommand{\rad}{\operatorname{rad}}

\newcommand{\Sp}{\mathbb{S}}

\newcommand{\eu}[1]{\left\lVert#1\right\rVert_2}

\newcommand{\skp}[2]{\langle #1, #2\rangle}

\newcommand{\eps}{\varepsilon}

\newcommand{\del}{\delta}
\newcommand{\Pb}{\mathbb{P}}
\newcommand{\E}{\mathbb{E}}

\newcommand{\D}{\mathcal{D}}

\newcommand{\sign}{\operatorname{sign}}
\newcommand{\la}{\lambda}
\newcommand{\ta}{\tau}
\newcommand{\tae}{\sigma}

\numberwithin{equation}{section}

\begin{document}

\title[Binarized Johnson-Lindenstrauss embeddings]{Binarized Johnson-Lindenstrauss embeddings}

\author{Sjoerd Dirksen}
\address{Utrecht University, Mathematical Institute, The Netherlands}
\email{s.dirksen@uu.nl}

\author{Alexander Stollenwerk}
\address{UCLouvain, ICTEAM Institute, ISPGroup,  Belgium}
\email{alexander.stollenwerk@uclouvain.be}

\maketitle

\begin{abstract}
We consider the problem of encoding a set of vectors into a minimal number of bits while preserving information on their Euclidean geometry. We show that this task can be accomplished by applying a Johnson-Lindenstrauss embedding and subsequently binarizing each vector by comparing each entry of the vector to a uniformly random threshold. Using this simple construction we produce two encodings of a dataset such that one can query Euclidean information for a pair of points using a small number of bit operations up to a desired additive error - Euclidean distances in the first case and inner products and squared Euclidean distances in the second. In the latter case, each point is encoded in near-linear time. The number of bits required for these encodings is quantified in terms of two natural complexity parameters of the dataset - its covering numbers and localized Gaussian complexity - and shown to be near-optimal.
\end{abstract}

\section{Introduction}

In modern data analysis one is frequently confronted with datasets that not only contain a large number of data points, but in addition each point is represented by a very high-dimensional vector. High-dimensional data imposes challenges in terms of storage consumption, computational expense during data analysis, and transmission of data across computing clusters. To alleviate these burdens, dimensionality reduction techniques have been applied in various areas, including clustering \cite{MMR19}, computational geometry \cite{Indyk01}, privacy \cite{BlockiBDS12}, spectral graph theory \cite{SpielmanS11}, linear programming \cite{CLAD20,LeS14}, numerical linear algebra \cite{TYUC19,Woo14}, multitask learning \cite{WDLSA09}, computational learning theory \cite{BalcanB05,BalcanBV06}, manifold learning \cite{Clarkson08,BHW07}, motif-finding in computational biology \cite{BuhlerT02}, and astronomy \cite{ContrerasM12}. Dimensionality reduction methods map a given dataset into a lower-dimensional space, while preserving essential structure for a target application. Although the structure to be preserved is application-specific, it is frequently desired to preserve the Euclidean geometry of the dataset, i.e., to preserve inter-point Euclidean distances, inner products, and angles. That is, for a given point set $\mathcal{D}\subset\R^n$ with $n$ very large, one desires to find a map $f:\mathcal{D}\rightarrow\R^m$, $m\ll n$, such that for any two query points $x,y\in \mathcal{D}$ one can accurately and quickly evaluate $\|x-y\|_2,\langle x,y\rangle$, and $\arccos(\langle x,y\rangle/(\|x\|_2\|y\|_2))$ using their images $f(x)$ and $f(y)$.\par
A classical tool is the Johnson-Lindenstrauss lemma \cite{JL84}, which states that Euclidean dimensionality reduction can be achieved using a random linear map. In its modern formulation, it states that if $A\in \R^{m\times n}$ has independent, $K$-subgaussian entries with mean zero and unit variance and $f:\R^n\to\R^m$ is defined by $f(x)=\frac{1}{\sqrt{m}}Ax$, then for any finite set $\mathcal{D}$ and $\eps>0$, 
\begin{equation}
\label{eqn:J-L}
(1-\eps)\|x-y\|_2^2 \leq \|f(x)-f(y)\|_2^2 \leq (1+\eps)\|x-y\|_2^2, \qquad \text{for all } x,y\in \mathcal{D} 
\end{equation}   
with high probability, provided that $m\gtrsim \eps^{-2}\log|\mathcal{D}|$, where $|\mathcal{D}|$ denotes the number of points in $\mathcal{D}$ and $a\lesssim b$ means that $a\leq Cb$ for a constant $C>0$ that only depends on the subgaussian constant $K$.\par     
The Johnson-Lindenstrauss lemma is remarkable in at least two respects. First, the dimensionality reducing map $f$ is \emph{data-oblivious}, i.e., it is constructed without any prior knowledge of the dataset to be embedded. This property is crucial in certain applications, e.g., one-pass streaming applications \cite{ClarksonW09} and data structural problems such as nearest neighbor search \cite{HarPeledIM12}. Second, the Johnson-Lindenstrauss embedding is in general \emph{optimal}: Larsen and Nelson \cite{LaN17} showed under the condition $\eps>\min\{n,|\mathcal{D}|\}^{-0.49}$ that any map $f:\mathcal{D}\to \R^m$ satisfying \eqref{eqn:J-L} must satisfy $m\gtrsim \eps^{-2}\log|\mathcal{D}|$. Motivated by these remarkable properties and the ubiquity of Euclidean dimensionality reduction in applications, a steady stream of works has strived to improve the construction of the classical Johnson-Lindenstrauss embedding and to refine its theoretical guarantees.\par 
In the first vein, several works have focused on improving the speed of the Johnson-Lindenstrauss embedding: if $A$ has independent subgaussian entries, then it is densely populated and hence computing $f(x)$ in general requires time $O(mn)$. The first work improving over this \cite{AC09} introduced the Fast Johnson-Lindenstrauss Transform, which achieved an embedding time $O(n\log n + m^3)$. An improved construction reduced this to time $O(n\log n + m^{2+\gamma})$ for any small constant $\gamma>0$ \cite{AL09}. More recently several works introduced transformations with $O(n\log n)$ embedding time at the expense of increasing the embedding dimension $m$ by a $(\log n)^c$ factor \cite{AL13,KW11,NPW14}. In all these works the dense matrix $A$ is replaced by a number of very sparse random matrices and discrete Fourier matrices. The improved embedding time stems from the Fast Fourier Transform (FFT) \cite{CooleyT65}. A second line of work, initiated in \cite{Achlioptas03} and strongly advanced in \cite{DKS10}, improved the embedding time by making $A$ sparse. By drawing $A$ from a distribution over matrices having at most $s$ non-zeroes per column, $Ax$ can be computed in time $O(sn)$. The best known construction in this direction is the sparse Johnson-Lindenstrauss Transform of \cite{KN14}, which achieves \eqref{eqn:J-L} with high probability under the conditions $m \gtrsim \eps^{-2}\log |\mathcal{D}|$ and $s \gtrsim \eps^{-1}\log |\mathcal{D}|$. This bound on $s$ is known to be optimal in general up to a factor $O(\log(1/\eps))$ for any linear map satisfying \eqref{eqn:J-L} \cite{NN13a}.\par
In the second vein, a number of works have refined the estimates on the embedding dimension in the Johnson-Lindenstrauss lemma. Although the bound $m\gtrsim \eps^{-2}\log|\mathcal{D}|$ is worst-case optimal, it is possible to prove refined, \emph{instance-optimal} bounds in which the factor $\log|\mathcal{D}|$ is replaced by quantities that measure the \emph{complexity} or \emph{intrinsic dimension} of the set. The first result in this direction, derived by Gordon \cite{Gordon88}, states that if $A$ has i.i.d.\ Gaussian entries, then $f(x)=\frac{1}{\sqrt{m}} Ax$ satisfies \eqref{eqn:J-L} with high probability if $m \gtrsim \eps^{-2}\ell_*(\mathcal{D_{\operatorname{nc}}})^2$, where $\mathcal{D}_{\operatorname{nc}} = \{(x-y)/\|x-y\|_2 \ : \ x,y\in \mathcal{D}\}$ denotes the set of normalized cords associated with $\mathcal{D}$ and 
\begin{equation}
\label{eqn:GaussComplexity}
\ell_*(T) = \E \sup_{x\in T} |\langle g,x\rangle|,
\end{equation}
where $g$ is a standard Gaussian vector, denotes the \emph{Gaussian complexity}. This result was later extended to matrices with independent subgaussian rows \cite{Dir16,KM05,MPT07}. It is easy to see that $\ell_*(\mathcal{D_{\operatorname{nc}}})^2$ is bounded by $\log|\mathcal{D}|$ for any finite set, but may be much smaller if the set has a low-complexity structure. For instance, if $\mathcal{D}$ is a set of (even infinitely many) $k$-sparse vectors, then $\ell_*(\mathcal{D_{\operatorname{nc}}})^2\lesssim Ck\log(en/k)$ (see \cite{Dir16} for many more examples of low-complexity sets). More recently, similar instance-optimal bounds have been found for a large class of fast Johnson-Lindenstrauss transforms \cite{ORS15} and the sparse Johnson-Lindenstrauss transform \cite{BDN15}.

\subsection{Main results}

Our paper fits into a general line of work \cite{AlK17,DM18,DiS18,HuS20,Jac15,Jac17,JaCa17,JLB13,OyR15,PlV14} developing a further enhancement of the Johnson-Lindenstrauss lemma: whereas all aforementioned embeddings encode every datum into a vector in $\R^m$, this new research line seeks to encode each datum into a minimal number of \emph{bits}. Specifically, one would like to find an embedding map $f:\mathcal{D}\to\{-1,1\}^B$ and reconstruction map $d:\{-1,1\}^B\times \{-1,1\}^B\to \R$ such that for any pair $x,y\in\mathcal{D}$, $d(f(x),f(y))$ is an accurate proxy of $\|x-y\|_2$, $\langle x,y\rangle$, or $\arccos(\langle x,y\rangle/(\|x\|_2\|y\|_2))$. Ideally, the bit complexity $B$ is minimal and both the embedding time (i.e., the time needed to compute $f(x)$) and the query time (the time to compute $d(u,v)$) are low. In this paper we show that surprisingly good results in this direction can be obtained by applying a traditional Johnson-Lindenstrauss embedding followed by a very simple binarization operation, in which each vector entry is compared to a random threshold. That is, we use maps of the form 
\begin{equation}
\label{eqn:GaussBinEmdDefIntro}
f: \R^n \to \{-1,1\}^m, \qquad x\mapsto f(x)=\sign(Ax+\ta),
\end{equation}
where $A\in \R^{m\times n}$ is a random linear embedding, $\tau\in \R^m$ is a random vector and the sign-function is applied component-wise. We quantify the sufficient number of bits in terms of two complexity measures: the localized version $\ell_*((\mathcal{D}-\mathcal{D})\cap \eps B_2^n)$ of the Gaussian complexity given in \eqref{eqn:GaussComplexity} and the covering number $\mathcal{N}(\mathcal{D},\eps)$, i.e., the smallest number of Euclidean balls with radius $\eps$ needed to cover $\mathcal{D}$.\par  
In our first result we aim to encode a dataset into a minimal number of bits, while preserving information on the Euclidean distances between points. Below, $d_H$ denotes the Hamming distance.
\begin{thm}
\label{thm:GaussianDithered}
Fix $R>0$ and $\D\subset RB^n_2$. Let $\del\in (0,\frac{R}{2}]$ and suppose that
\begin{align}
\label{eqn:GaussianDitheredBoundm}
& \la\gtrsim R\sqrt{\log(R/\del)}, \nonumber \\
& m\gtrsim \la^2\del^{-2} \log\mathcal{N}(\D,\eps) +  \la \del^{-3} \ell_*((\D-\D)\cap \eps B^n_2)^2, 
\end{align}
where $\eps \lesssim \del/\sqrt{\log(e\la/\del)}$. Let $A\in \R^{m\times n}$ be standard Gaussian and let $\ta\in [-\la,\la]^m$ be uniformly distributed and independent of $A$.  Then, with probability at least $1-2 \exp(-c\del^2 m/\la^2)$, the map $f$ in \eqref{eqn:GaussBinEmdDefIntro} satisfies
\begin{equation}
\label{eqn:GaussianDithered}
\sup_{x,y\in \D}\Big|\frac{\sqrt{2\pi}\la}{m}\hd{f(x)}{f(y)}-\|x-y\|_2\Big|\leq  \del.
\end{equation}
\end{thm}
This result has a geometric interpretation in terms of the hyperplanes $H_{a_i,\tau_i} = \{ z \in \R^n : \langle a_i,z\rangle +\tau_i=0\}$ that are generated by the rows $a_i$ of $A$: \eqref{eqn:GaussianDithered} means that $\lambda\sqrt{2\pi}$ times the fraction of the hyperplanes that separate $x$ and $y$ approximates the Euclidean distance between $x$ and $y$ up to an additive error $\delta$. The embedding in Theorem~\ref{thm:GaussianDithered} enjoys a very low query time (requiring only the comparison of two bit strings of length $m$) and, as we will discuss below, a near-optimal bit complexity. On the downside, the embedding time is $O(mn)$.\par 
In our second main result, we use a different construction featuring a fast Johnson-Lindenstrauss embedding to bring the embedding time down to $O(n\log n)$, while maintaining a low query time and bit complexity. In this result, however, our reconstruction map instead approximates \emph{squared} Euclidean distances. We consider the random circulant matrix $A=R_I\Gamma_{\xi}D_{\theta}$. Here, $I\subset [n]$ is a fixed subset with $|I|=m$, $\Gamma_{\xi}$ is the circulant matrix generated by a mean-zero random vector $\xi\in \R^n$ with independent, unit variance, $K$-subgaussian entries and $D_{\theta}$ is a diagonal matrix containing independent Rademachers on its diagonal. To state our result, we use the following notation. For $\la>0$, we consider the binary embedding map 
\begin{equation}
\label{eqn:defbigF}
F:\R^n\to \{-1,1\}^{2m}, \quad x\mapsto F(x)= \begin{bmatrix}
\sign(Ax+\ta)\\
\sign(Ax+\ta')
\end{bmatrix}
\end{equation}
where $\ta,\ta'$ are uniformly distributed in $[-\la,\la]^m$ and $A$, $\tau$, and $\tau'$ are independent. We let
\begin{equation*}
\skp{\;}{\,}_{S_m}: \R^{2m}\times \R^{2m}\to \R , \quad \skp{a}{b}_{S_m}:=\skp{a}{S_m b},\quad a,b\in \R^{2m},
\end{equation*}
denote the (indefinite) symmetric bilinear form on $\R^{2m}$ induced by the $2m\times 2m$ matrix 
\begin{equation*}
S_m=
\begin{pmatrix}
0 & \operatorname{Id}_{m} \\ 
\operatorname{Id}_{m} &  0
\end{pmatrix}.
\end{equation*}
\begin{thm}
\label{thm:Gaussian_Circulant_Dithered}
	There exists an absolute constant $c>0$ and a polylogarithmic factor $\alpha$ satisfying $\alpha\leq \log^4(n)+\log(\eta^{-1})$ such that the following holds.
	Let $\D\subset RB^n_2$ for $R\geq 1$ and fix $\eta\in (0,\frac{1}{2}]$. For any $\del\in (0,1)$ and $\lambda\geq R$, if
	$$\lambda\gtrsim \alpha R\sqrt{\log(e\la^2/\del R^2)}, \qquad R^2\geq \del \lambda^2,$$
	and
	\begin{equation}
	\label{eqn:GaussianCirculantDitheredBoundm}
		m\gtrsim  \alpha^2 \del^{-2} \log(\mathcal{N}(\D,r))+ \alpha^2 \lambda^{-2}\del^{-3}\ell_*((\D-\D)\cap r B^n_2)^2, \\
	\end{equation}
	for a parameter $r\leq c\del R$,
	then with probability at least $1-\eta$
	\begin{equation*}\label{eq:fast_inner_product}
	\sup_{x,y\in \D}\left|\frac{\lambda^2}{2m}\skp{F(x)}{F(y)}_{S_m} - \skp{x}{y}\right|\leq \del \la^2
	\end{equation*}
	and, as a consequence,
	\begin{equation*}\label{eq:fast_Euclidean_distance}
	\sup_{x,y\in \D}\left|\frac{\lambda^2}{2m}\skp{F(x)-F(y)}{F(x)-F(y)}_{S_m}-\|x-y\|_2^2\right|\leq 4\del \la^2.
	\end{equation*}
\end{thm}

\subsection{Near-optimality of the results}

Let us briefly discuss the near-optimality of the bit complexity in our two main results. Alon and Klartag \cite{AlK17} recently investigated the minimal number $B(N,n,\delta)$ of bits required to encode a set of $N$ points in $B_2^n$ into a data structure so that one can recover pairwise squared Euclidean distances up to an additive error $\delta$. They determined $B(N,n,\delta)$ for any $N\geq n\geq 1$ and $\delta^{-0.49}<N$. In particular, they showed that $B(N,n,\delta)$ is equivalent (up to universal constants) to $\delta^{-2}N\log N$ if $n\geq\delta^{-2}\log N$. We adapt their argument to give a lower bound on the number of bits that is required by any oblivious, possibly random binary embedding that is allowed to fail with probability $\eta$. Assuming $n\gtrsim \del^{-2}\log(N/\eta)$ and $N/\eta \gtrsim \del^{-4}\log^2(N/\eta)$, we show that for any random $f: B^n_2\to \{-1,1\}^m$ and $d:\{-1,1\}^m\times \{-1,1\}^m\to \R$ satisfying 	
\begin{equation}
\label{eqn:probEmbedFail}
	\bP\Big(\sup_{x,y\in \D}|d(f(x),f(y))- \eu{x-y}|\leq \del\Big)\geq 1-\eta
\end{equation}
for any $\D\subset B^n_2$ with $|\D|=N$, we must have $m\gtrsim \del^{-2} \log(N/\eta)$. The same result holds if the Euclidean distances in \eqref{eqn:probEmbedFail} are replaced by their squares. This shows that the bit complexity in Theorem~\ref{thm:GaussianDithered} is near-optimal for general finite sets. Indeed, if $\mathcal{D}$ is a finite set in the Euclidean unit ball (i.e., $R=1$ so that we can take $\lambda\sim \log^{1/2}(1/\delta)$), then it is easily seen that 
$$\log(\mathcal{N}(\D,r))\leq \log|\mathcal{D}|, \qquad \ell_*((\D-\D)\cap r B^n_2)^2\lesssim r^2\log|\mathcal{D}|$$
and hence Theorem~\ref{thm:GaussianDithered} implies that \eqref{eqn:GaussianDithered} holds with probability at least $1-\eta$ for any $\D\subset B^n_2$ with $|\D|=N$ if $m\gtrsim \del^{-2}\log(1/\delta)\log(N/\eta)$. Similarly, Theorem~\ref{thm:Gaussian_Circulant_Dithered} is optimal up to logarithmic factors.\par
Under additional restrictions (in particular, that $f$ induces convex quantization cells), a more refined lower bound can be derived which features the complexity parameters in Theorems~\ref{thm:GaussianDithered} and \ref{thm:Gaussian_Circulant_Dithered}. Our main result in this direction, Theorem~\ref{thm:lowerbound:continuous}, shows in particular that $m$ needs to scale at least as 
$$m\sim \delta^{-2}\log\mathcal{N}(\D,\delta) + \del^{-2} \ell_*((\D-\D)\cap \delta B^n_2)^2$$ 
for the conclusions of Theorems~\ref{thm:GaussianDithered} and \ref{thm:Gaussian_Circulant_Dithered} to hold for $R=1$. Hence, our bounds on $m$ in these two results are optimal up to logarithmic factors and possibly the scaling $\delta^{-3}$ in front of the localized Gaussian complexity parameter in \eqref{eqn:GaussianDitheredBoundm} and \eqref{eqn:GaussianCirculantDitheredBoundm}.

\subsection{Related work} Several prior works have analyzed quantized John\-son-Lin\-den\-strauss embeddings, i.e., compositions of Johnson-Lindenstrauss embeddings with quantization schemes stemming from signal processing, to encode data vectors into bit strings. Apart from the interest in these embeddings for the purpose of dimensionality reduction, quantized Johnson-Lindenstrauss embeddings are tightly connected with compressed sensing with quantized measurements, see \cite{Dir19} for further details. We divide our review of the literature according to the quantization scheme used. To facilitate the comparison of the various bit complexity estimates with \eqref{eqn:GaussianDitheredBoundm} and \eqref{eqn:GaussianCirculantDitheredBoundm}, we recall that $\log\mathcal{N}(\D,\eps)\lesssim \eps^{-2}\ell_*(\mathcal{D})^2$ for any $\eps>0$ by Sudakov's inequality and note that trivially $\ell_*((\D-\D)\cap \eps B^n_2)\leq \eps\ell_*(\mathcal{D}_{\operatorname{nc}})$. 
\par 
\vspace{0.1cm}
\textbf{Memoryless one-bit quantization.} The works \cite{DM18,DiS18,JLB13,OyR15,PlV14} studied quantized Johnson-Lindenstrauss embeddings featuring the memoryless one-bit quantizer $\sign(\cdot+\tau)$ featured in our present paper. A beautiful embedding result for general subsets of the unit sphere was derived by Plan and Vershynin \cite{PlV14}. They showed that if $\mathcal{D}\subset S^{n-1}$, $m\gtrsim \delta^{-6} \ell_*(\mathcal{D})^2$, and $A \in \R^{m \times n}$ is a standard Gaussian matrix then, with probability at least $1-2e^{-cm\delta^2}$, for all $x,y\in \mathcal{D}$,
\begin{equation}
\label{eqn:PlVEmbed}
\Big|\frac{1}{m}d_H(\operatorname{sign}(Ax),\operatorname{sign}(Ay)) - d_{S^{n-1}}(x,y)\Big|\leq \del,
\end{equation}
where $d_{S^{n-1}}=\tfrac{1}{\pi}\arccos(\langle x,y\rangle/(\|x\|_2\|y\|_2))$ denotes the angular distance. It was later shown in \cite{OyR15} that \eqref{eqn:PlVEmbed} remains true if $m\gtrsim \delta^{-4} \ell_*(\mathcal{D})^2$. Moreover, for certain `simple' sets (e.g., if $\mathcal{D}$ is the set of unit norm sparse vectors) it is known that $m\gtrsim \delta^{-2} \ell_*(\mathcal{D})^2$ suffices for \eqref{eqn:PlVEmbed} (see \cite{JLB13,OyR15,PlV14} for examples).\par
In \cite{DM18} it was shown that one can embed general sets in a Euclidean ball (instead of the unit sphere) and approximate Euclidean distances (rather than angular distances) by taking uniformly distributed thresholds in the quantizer. Concretely, suppose that $A$ has i.i.d.\ symmetric, isotropic, $K$-subgaussian rows and that the entries of $\tau$ are i.i.d.\ uniformly distributed on $[-\la,\la]$. If $\mathcal{D} \subset R B_2^n$, $\la = c_0 R$ and
\begin{equation}
\label{eqn:bitCompDM18}
m \geq c_1\frac{R \log(eR/\delta)}{\delta^3} \ell_*(\mathcal{D})^2,
\end{equation}
then with probability at least $1-8\exp(-c_2m\delta/R)$, for any $x,y$ in the convex hull $\operatorname{conv}(\mathcal{D})$ of $\mathcal{D}$ satisfying $\|x-y\|_2 \geq \delta$, one has
\begin{equation} \label{eq:isomorphic}
c_3\frac{ \|x-y\|_2}{R} \leq \frac{1}{m}d_H(\operatorname{sign}(Ax+\tau),\operatorname{sign}(Ay+\tau)) \leq c_4\sqrt{\log(eR/\delta)} \cdot \frac{\|x-y\|_2}{R},
\end{equation}
where $c_0,\ldots,c_4$ depend only on $K$. In contrast to \eqref{eqn:GaussianDithered} and \eqref{eqn:PlVEmbed}, which are near-isometric bounds, \eqref{eq:isomorphic} is an isomorphic bound: Euclidean distances are only approximated up to a constant/logarithmic factor. Although one cannot expect a near-isometric estimate to hold for general subgaussian matrices, Theorem~\ref{thm:GaussianDithered} shows that this is possible if $A$ is Gaussian, under a near-optimal bit complexity.
\par     
\vspace{0.1cm}
\textbf{Uniform scalar quantization.} The works \cite{Jac15,Jac17,JaCa17} analyzed quantized Johnson-Lindenstrauss embeddings related to the uniform scalar quantizer $Q_{\rho,\tau}:\R^m\rightarrow (\rho\mathbb{Z})^m$ defined by 
$$Q_{\rho,\tau}(z) = \big(\rho\lfloor (z_i+\tau_i)/\rho\rfloor\big)_{i=1}^m.$$
Geometrically, $Q_{\rho,0}$ divides $\R^m$ into half-open cubes with side lengths equal to $\rho$ and maps any vector $z\in \R^m$ to the corner of the cube in which it is located. In \cite{JaCa17} it is shown that if $\tau$ is uniformly distributed in $[-\rho,\rho]^m$, $m\gtrsim \epsilon^{-2}\log \mathcal{N}(\mathcal{D},\rho\epsilon^2)$ and 
\begin{equation}
\label{eqn:RIP12def}
(1-\nu)\|z\|_2\leq \frac{1}{m}\|Az\|_1\leq (1+\nu)\|z\|_2, \qquad \text{for all } z\in \mathcal{D}-\mathcal{D},
\end{equation}
then with probability at least $1-C e^{-cm\epsilon^2},$ 
\begin{equation}
\label{eqn:QJL}
\Big|\frac{1}{m}\|Q_{\rho,\tau}(Ax)-Q_{\rho,\tau}(Ay)\|_1-\|x-y\|_2\Big|\leq \nu\|x-y\|_2 + c\rho\epsilon
\end{equation}
for all $x,y\in \mathcal{D}$. A special case of a result of Schechtman \cite{Sch06} shows that if $B$ is standard Gaussian and $A=\sqrt{\frac{\pi}{2}}B$, then $A$ satisfies \eqref{eqn:RIP12def} with probability at least $1-2e^{-m\nu^2/2}$ if $m\gtrsim \nu^{-2} \ell_*(\mathcal{D}_{\text{nc}})^2$ (see also \cite[Lemma 2.1]{PlV14} for a short proof of this special case). Note that in these results $m$ does not represent the bit complexity: it will depend on the $\ell_{\infty}$-diameter of $A\mathcal{D}$ and the choice of $\rho$. If $\mathcal{D}\subset B_2^n$, then the bit complexity is with high probability $O(m\log(m/\rho))$.
\par 
\vspace{0.1cm}
\textbf{Noise shaping methods.} The recent work \cite{HuS20} considered the use of noise shaping methods, a family of adaptive quantization methods that include the popular sigma-delta quantization method. Let $Q$ be a stable one-bit distributed noise-shaping quantizer and let $V\in \R^{p\times m}$ be the normalized condensation matrix associated with this scheme (see \cite{HuS20} for precise definitions and further details). For the purpose of this discussion, it is sufficient to know that $Q(z)$ and $Vz$ can be computed in time $O(m)$. Let $A=R_I\Gamma_{\xi}D_{\theta}$ be the subgaussian circulant matrix featured in Theorem~\ref{thm:Gaussian_Circulant_Dithered}, let $D_{\sigma}$ be a diagonal matrix with independent Rademachers on its diagonal that is independent of $A$ and let $g(z)=Q(8D_{\sigma}Az/9)$. The main result \cite[Theorem 5.8]{HuS20} states the following: if $\mathcal{D}\subset B_1^n$ and 
$$m\geq p\gtrsim \nu^{-2}\log^2(1/\eta)\log^4(n)\max\Big\{1,\frac{\ell_*(\mathcal{D}-\mathcal{D})^2}{\operatorname{rad}(\mathcal{D}-\mathcal{D})^2}\Big\}$$
then with probability at least $1-\eta$,
$$|\|Vg(x)-Vg(y)\|_2-\|x-y\|_2|\leq \max\{\sqrt{\nu},\nu\} \operatorname{rad}(\mathcal{D}-\mathcal{D}) + c^{-m/p}$$
for all $x,y\in \D$.
Assuming the nontrivial case where $\operatorname{rad}(\mathcal{D}-\mathcal{D})\geq \delta$, we can set $\nu=\delta^2/ \operatorname{rad}(\mathcal{D}-\mathcal{D})^2\leq 1$ and take $m/p\sim \log(1/\delta)$ to find that
$$|\|Vg(x)-Vg(y)\|_2-\|x-y\|_2|\leq\delta$$
when
$$m\geq \delta^{-4}\log(1/\delta)\log^2(1/\eta)\log^4(n)\max\Big\{\operatorname{rad}(\mathcal{D}-\mathcal{D})^4,\operatorname{rad}(\mathcal{D}-\mathcal{D})^2\ell_*(\mathcal{D}-\mathcal{D})^2\Big\}.$$

\subsection{Organization}
Our paper is structured as follows: in Section~\ref{sec:proofThmMain1} we prove Theorem~\ref{thm:GaussianDithered}, Section~\ref{sec:proofThmMain2} is devoted to the proof of Theorem~\ref{thm:Gaussian_Circulant_Dithered}, and in the final section we derive lower bounds on the best possible bit complexity. 

\subsection{Notation}

For an integer $n\in \N$, we set $[n]=\{1,\ldots,n\}$. For any $t\in \R$, $\lfloor t\rfloor$ is the largest integer that is smaller than $t$. For a given set $A$ we denote its cardinality by $|A|$. For $p\geq 1$, the $\ell_p$-norm of a vector $x\in \R^n$ is denoted by
$\|x\|_p$ and the associated unit ball is $B^n_p$.
Given $x\in \R^n$ we define $D_x\in \R^{n\ti n}$ to be the diagonal matrix obtained from $x$, i.e., $(D_x)_{i,j}=\del_{i,j}x_i$. $\Gamma_x$ denotes the $n\times n$ circulant matrix generated by $x\in \R^n$, that is, $(\Gamma_{x})_{i,j}=x_{j-i \, \text{mod} \, n}$. A random variable $X$ is called $K$-subgaussian if $\|X\|_{\psi_2}=\inf\{t>0\; : \; \E\exp(X^2/t^2)\leq 2\}\leq  K$.  More generally, a random vector $X\in \R^n$ is $K$-subgaussian if 
$\|X\|_{\psi_2}=\sup_{x\in B^n_2}\|\skp{X}{x}\|_{\psi_2}\leq K$. A random variable is called Rademacher if it takes values $1$ and $-1$ with equal probability. A random vector is called Rademacher if it has independent Rademacher entries. For $k\in [n]$ and $x\in \R^n$, we denote the $\ell_2$-norm of the $k$ in magnitude largest entries of $x$ by
$$\|x\|_{[k]} = \Big(\sum_{i=1}^k (x^*_i)^2\Big)^{1/2} = \sup_{|I|\leq k} \Big(\sum_{i\in I} x_i^2\Big)^{1/2},$$
where $x^{\ast}\in \R^n$ denotes the nonincreasing rearrangement of $(|x_i|)_{i\in [n]}$.
The Hamming distance between $a,b\in \{-1,1\}^m$ is given by $\hd{a}{b}=|\{i\in [m]\; : \; a_i\neq b_i\}|$.
We use $C,c>0$ to denote constants that may depend only on the subgaussian constants involved. Their value might change from line to line.
We write $a\lesssim b$ if $a\leq Cb$ for a constant $C>0$ that only depends on subgaussian constants, and we use the abbreviation $a\simeq b$ if both $a \lesssim b$ and $b\lesssim a$ hold (with possibly different implicit constants).

\section{Proof of Theorem~\ref{thm:GaussianDithered}}
\label{sec:proofThmMain1}

Throughout this subsection, we consider an $m\times n$ standard Gaussian matrix $A$, and for a parameter
$\la>0$ we let $\ta\in [-\la,\la]^m$ denote a uniformly distributed random vector which is independent of $A$. 
Further, we consider the map 
\begin{equation}
\label{eqn:GaussBinEmdDef}
f: \R^n \to \{-1,1\}^m, \qquad x\mapsto f(x)=\sign(Ax+\ta),
\end{equation}
where the sign-function is applied component-wise. Finally, we use shorthand-notation for a rescaled version of the Hamming distance on $\{-1,1\}^m$, namely 
\begin{equation*}
d_\la(b,c)=\frac{\sqrt{2\pi}\la}{m}\hd{b}{c}.
\end{equation*}
To prepare for the proof of Theorem~\ref{thm:GaussianDithered} we collect several useful observations. 
\begin{lemma}\label{lem:exp} 
Fix $\lambda>0$ and let $\tae$ be uniformly distributed in $[-\la, \la]$. Define the function 
\begin{equation*}
\phi_\la(s)=(|s|-\la)1_{|s|\geq \la} \qquad \text{for } s\in \R.
\end{equation*}
Then for $a,b\in \R$,
\begin{align*}
\big|2\la \bP(\sgn(a+\tae)\neq \sgn(b+\tae))- |a-b|\big|\leq \phi_\la(a)+\phi_\la(b).
\end{align*}
\end{lemma}

\begin{proof}
We may assume that $a<b$. If $|a|\leq \la$ and $|b|\leq \la$ then $\phi_\la(a)+\phi_\la(b)=0$. Hence, in this case we have to show that
\begin{equation*}
2\la\bP(\sgn(a+\tae)\neq \sgn(b+\tae))=|a-b|. 
\end{equation*}
To see this, note that
\begin{equation*}
\bP(\sign(a+\tae)\neq \sign(b+\tae))=\bP(a<-\tae\leq b)=(b-a)/2\la,
\end{equation*}
as $-\tae$ is uniformly distributed in $[-\la,\la]$.
Now assume that $|a|>\la$ or $|b|>\la$. We 
distinguish three cases according to the value of $a$:
\begin{itemize}
\item [1)] \underline{$a<-\la$.} We distinguish three sub-cases according to the value of $b$:\\
If $b<-\la$ then $2\la\bP(\sgn(a+\tae)\neq \sgn(b+\tae))=0$, and 
\begin{align*}
|2\la\bP(\sgn(a+\tae)\neq \sgn(b+\tae))- |a-b||=|a-b|&\leq |a+\la|+|b+\la|\\
&=-(a+\la)-(b+\la)\\
&=\phi_\la(a)+\phi_\la(b).
\end{align*}
If $|b|\le \la$ then $\sign(a+\tau)=-1$ implies 
$$2\la\bP(\sgn(a+\tae)\neq \sgn(b+\tae))=2\la\bP(b\geq -\tae)=(b-(-\la))=b+\la$$ 
and so
\begin{align*}
|2\la\bP(\sgn(a+\tae)\neq \sgn(b+\tae))- |a-b||&=| b+\la -|a-b| |\\
&=| b+\la -(b-a) |\\
&=\phi_\la(a)\leq \phi_\la(a)+\phi_\la(b).
\end{align*}
If $\la<b$ then $2\la\bP(\sgn(a+\tae)\neq \sgn(b+\tae))=2\la$, and 
\begin{align*}
|2\la\bP(\sgn(a+\tae)\neq \sgn(b+\tae))- |a-b||&=|2\la - (b-a)|\\
&=(b-a)-2\la\\
&=-a-\la + b-\la\\
&=\phi_\la(a)+\phi_\la(b).
\end{align*}
\item [2)] \underline{$|a|\leq \la$.} Then $\la<b$ and 
$$2\la\bP(\sgn(a+\tae)\neq \sgn(b+\tae))=2\la\bP(a<-\tae) = \la-a.$$
Hence
\begin{align*}
|2\la\bP(\sgn(a+\tae)\neq \sgn(b+\tae))- |a-b||&=| \la-a  - (b-a)  |\\
&=|\la-b|=b-\la\\
&=|b|-\la=\phi_\la(b)\leq\phi_\la(a)+\phi_\la(b).
\end{align*}
\item [3)] \underline{$\la<a$.} Then $\la<b$ and 
$2\la\bP(\sgn(a+\tae)\neq \sgn(b+\tae))=0$. Hence
\begin{align*}
|2\la\bP(\sgn(a+\tae)\neq \sgn(b+\tae))- |a-b||=|a-b|&\leq |a-\la|+ |b-\la|\\
&=(|a|-\la)+ (|b|-\la)\\
&=\phi_\la(a)+\phi_\la(b).
\end{align*}
\end{itemize} 
\end{proof}
Our next lemma shows that $d_{\lambda}(f(x),f(y))$ is in expectation a good proxy for the distance between $x$ and $y$. 
\begin{lemma}\label{lem:bias_reduc}
Let $f:\R^n \to \{-1,1\}^m$ be the map defined in \eqref{eqn:GaussBinEmdDef}.
Then, for any $x,y\in \R^n$,
\begin{equation*}
\big|\E d_{\lambda}(f(x),f(y)) - \|x-y\|_2\big|\leq 2 r_{x,y}\exp(-\la^2/2r_{x,y}^2),
\end{equation*}
where $r_{x,y}=\max\{\eu{x},\eu{y}\}$.
\end{lemma}
\begin{proof}
Since every row $a_i$ of $A$ is standard Gaussian, 
$$\sqrt{\frac{\pi}{2}} \E |\skp{a_i}{x-y}|=\|x-y\|_2$$
for all $1\leq i\leq m$ and hence
\begin{equation*}
\E \sqrt{\frac{\pi}{2}} \frac{1}{m}\|A(x-y)\|_1=\frac{1}{m}\sum_{i=1}^m \sqrt{\frac{\pi}{2}} \E |\skp{a_i}{x-y}|=\eu{x-y}.
\end{equation*}
Set $\phi_\la(s)=(|s|-\la)1_{|s|\geq \la}$.
Using the independence of $A$ and $\tau$, Jensen's inequality and the triangle inequality, we obtain 
\begin{align*}
&\big|\E d_{\lambda}(f(x),f(y)) - \|x-y\|_2\big|\\
&\leq \sqrt{\frac{\pi}{2}} \E_A \Big(\frac{1}{m}\sum_{i=1}^m \big|2\la  \bP_\ta \big(\sign(\skp{a_i}{x}+\ta_i)\neq 
\sign(\skp{a_i}{y}+\ta_i)\big) - |\skp{a_i}{x-y}|\big| \Big)\\
&\leq \sqrt{\frac{\pi}{2}} \E_A \Big( \frac{1}{m}\sum_{i=1}^m [\phi_\la(\skp{a_i}{x})+ \phi_\la(\skp{a_i}{y})]\Big),
\end{align*}
where the last inequality follows from Lemma~\ref{lem:exp}. 
Again using that every vector $a_i$ is distributed as a standard Gaussian random vector $g$, we obtain
\begin{equation}
\label{eqn:biasRedEstimate}
|\E d_{\lambda}(f(x),f(y)) - \|x-y\|_2| \leq \sqrt{\frac{\pi}{2}} \E(\phi_\la(\skp{g}{x})+ \phi_\la(\skp{g}{y})).
\end{equation}
If $z= 0$, then $\E(\phi_\la(\skp{g}{z})=\phi_\la(0)=0$. Otherwise,
\begin{align}\label{eq:bound_expectation_gaussian}
\E \phi_\la(\skp{g}{z})\leq \E [|\skp{g}{z}|1_{|\skp{g}{z}|\geq \la}]
&= \sqrt{\frac{2}{\pi}}\eu{z}\int_{\la/\eu{z}}^{\infty} t \exp(-t^2/2)\,dt\\ \nonumber
&=\sqrt{\frac{2}{\pi}}\eu{z}\exp(-\lambda^2/2\eu{z}^2),
\end{align}
Using this with $z=x$ and $z=y$ in \eqref{eqn:biasRedEstimate} yields the result.
\end{proof}
\begin{thm} \label{thm:monotone-subgaussian}
There exist constants $c, C>0$ that only depend on $K$ such that the following holds. 
Let $A$ be an $m\times n$ matrix whose rows $a_1, \ldots, a_m$ are independent copies of an isotropic $K$-subgaussian random vector $a$. Let $T \subset \R^n$ and $k\in [m]$.
For any $u \geq 1$ with probability at least $1-2\exp(-c u^2 k \log(em/k))$,
$$
\sup_{z \in T} \|Az\|_{[k]} \leq C \left(\ell_*(T) + u \rad(T) \sqrt{k \log(em/k)} \right).
$$
\end{thm}
We omit the proof of Theorem \ref{thm:monotone-subgaussian}, which is standard. It is based on generic chaining (see e.g.\ \cite[Theorem 3.2]{Dir15}) combined with Talagrand's majorizing measures theorem \cite{Tal14}.\par
We are now ready to prove the main result of this section.\par
\begin{proof}[Proof of Theorem~\ref{thm:GaussianDithered}]
Let $N_\eps \subset \D$ be a minimal $\eps$-net of $\D$ with respect to the Euclidean metric. 
For any $x\in \D$ let $\pi(x)\in \text{argmin}_{z\in N_\eps}\|x-z\|_2$. By the triangle inequality, 
\begin{align}\label{eq:four_summands}
\big|d_{\lambda}(f(x),f(y))-\|x-y\|_2\big| \nonumber
&\leq \big|d_{\lambda}(f(x),f(y))- d_{\lambda}(f(\pi(x)),f(\pi(y)))\big|\\ \nonumber
&\qquad+\big|d_{\lambda}(f(\pi(x)),f(\pi(y))) - \E d_{\lambda}(f(\pi(x)),f(\pi(y)))\big|\\ \nonumber
&\qquad +\big|\E d_{\lambda}(f(\pi(x)),f(\pi(y))) - \|\pi(x)-\pi(y)\|_2\big|\\
&\qquad+\big|\|\pi(x)-\pi(y)\|_2-\|x-y\|_2\big|.
\end{align}
Clearly, the last summand is bounded by $2\eps$. Since $\eu{\pi(x)}\leq R$ for every $x \in \D$, 
Lemma~\ref{lem:bias_reduc} implies
\begin{equation*}
\sup_{x,y\in \D}\big|\E d_{\lambda}(f(\pi(x)),f(\pi(y))) - \|\pi(x)-\pi(y)\|_2\big|\leq 2R\exp(-\la^2/2R^2)\leq \del,
\end{equation*}
where the last inequality follows from $\la\gtrsim R\sqrt{\log(R/\del)}$.
In order to bound the second summand on the right hand side of \eqref{eq:four_summands}, we invoke Hoeffding's inequality which implies that for every $x,y\in \D$ and $\del>0$,
\begin{equation*}
\bP\big(\big|d_{\lambda}(f(\pi(x)),f(\pi(y))) - \E d_{\lambda}(f(\pi(x)),f(\pi(y)))\big|\geq \del\big)\leq 2 \exp(-c\del^2 m/\la^2),
\end{equation*}
where $c>0$ denotes an absolute constant. If $m\gtrsim \del^{-2}\la^2 \log\mathcal{N}(\D,\eps)$ then 
the union bound yields  
\begin{equation*}
\sup_{x,y\in \D}\big|d_{\lambda}(f(\pi(x)),f(\pi(y))) - \E d_{\lambda}(f(\pi(x)),f(\pi(y)))\big|\leq \del
\end{equation*}
with probability at least $1-2 \exp(-c\del^2 m/\la^2)$. 
To bound the first summand on the right hand side of \eqref{eq:four_summands} we first apply the triangle inequality: 
\begin{align*}
&\big|d_{\lambda}(f(x),f(y))- d_{\lambda}(f(\pi(x)),f(\pi(y)))\big|\\
&\leq\frac{\sqrt{2\pi}\la}{m}\sum_{i=1}^m \big|1_{\sign(\skp{a_i}{x}+\ta_i)\neq \sign(\skp{a_i}{y}+\ta_i)} -1_{\sign(\skp{a_i}{\pi(x)}+\ta_i)\neq \sign(\skp{a_i}{\pi(y)}+\ta_i)}\big|.
\end{align*}
Observe that if  
$\sign(\skp{a_i}{x}+\ta_i)=\sign(\skp{a_i}{\pi(x)}+\ta_i)$ 
and 
$\sign(\skp{a_i}{y}+\ta_i)=\sign(\skp{a_i}{\pi(y)}+\ta_i)$,
then 
$$1_{\sign(\skp{a_i}{x}+\ta_i)\neq \sign(\skp{a_i}{y}+\ta_i)} -1_{\sign(\skp{a_i}{\pi(x)}+\ta_i)\neq \sign(\skp{a_i}{\pi(y)}+\ta_i)}=0.$$
Consequently, 
\begin{align}\label{eq:applyDM18}
&\sup_{x,y\in \D}\big|d_{\lambda}(f(x),f(y))- d_{\lambda}(f(\pi(x)),f(\pi(y)))\big|\\ \nonumber
&\qquad \leq 2\sup_{x\in \D} \frac{\sqrt{2\pi}\la}{m}\sum_{i=1}^m 1_{\sign(\skp{a_i}{x}+\ta_i)\neq \sign(\skp{a_i}{\pi(x)}+\ta_i)}.
\end{align}
In addition,
\begin{equation}
\label{eqn:indZeroSep}
1_{\sign(\skp{a_i}{x}+\ta_i)\neq \sign(\skp{a_i}{\pi(x)}+\ta_i)}=0
\end{equation}
on 
$$A_{\delta} = \{ |\langle a_i,\pi(x)\rangle+\tau_i|>\delta\geq |\langle a_i,x-\pi(x)\rangle|\}.$$
As a consequence, 
\begin{equation}
\label{eqn:signChangeCov}
1_{\sign(\skp{a_i}{x}+\ta_i)\neq \sign(\skp{a_i}{\pi(x)}+\ta_i)}\leq 1_{A_{\delta}^c} \leq 1_{|\langle a_i,\pi(x)\rangle+\tau_i|\leq \delta} + 1_{|\langle a_i,x-\pi(x)\rangle|>\delta}.
\end{equation}
We conclude that 
\begin{align}\label{eq:proof:Gaussian:twoprocesses}
&\sup_{x,y\in \D}\big|d_{\lambda}(f(x),f(y))- d_{\lambda}(f(\pi(x)),f(\pi(y)))\big|\\ 
&\qquad \leq 2\sup_{y\in N_{\eps}} \frac{\sqrt{2\pi}\la}{m}\sum_{i=1}^m 1_{|\langle a_i,y\rangle+\tau_i|\leq \delta} 
  + 2\sup_{z\in (\D-\D)\cap \eps B_2^n} \frac{\sqrt{2\pi}\la}{m}\sum_{i=1}^m 1_{|\langle a_i,z\rangle|>\delta}\nonumber.
\end{align}
Since 
$$\bP_{\tau}(|\langle a_i,y\rangle+\tau_i|\leq \delta)\leq \frac{\delta}{\lambda},$$
the Chernoff bound implies that 
$$\sum_{i=1}^m 1_{|\langle a_i,y\rangle+\tau_i|\leq \delta}\leq \frac{2\delta m}{\lambda}$$
with probability at least $1-\exp(-cm\delta/\lambda)$. Therefore, if $m\gtrsim \lambda \delta^{-1} \log|N_\eps|$
then the union bound implies that 
$$\sup_{y\in N_{\eps}} \frac{\sqrt{2\pi}\la}{m}\sum_{i=1}^m 1_{|\langle a_i,y\rangle+\tau_i|\leq \delta} \lesssim \delta$$
with probability at least $1-\exp(-c'm\delta/\lambda)$. 
To estimate the second term on the right hand side of \eqref{eq:proof:Gaussian:twoprocesses}, first
observe that 
\begin{equation}
\label{eqn:largeCoeffk-norm1}
\sup_{z\in (\D-\D)\cap \eps B_2^n}|\{i \in [m] \ : \ |\skp{a_i}{z}|> \del\}|\leq \Big\lfloor\frac{\del m}{\lambda}\Big\rfloor
\end{equation}
if and only if $(Az)^*_{\lfloor\delta m/\lambda\rfloor}\leq \delta$ for all $z\in (\D-\D)\cap \eps B_2^n$. Clearly, 
\begin{equation}
\label{eqn:largeCoeffk-norm2}
(Az)^*_{\lfloor\delta m/\la\rfloor} \leq \Big(\frac{\la}{\delta m} \|Az\|_{[\lfloor\delta m/\la\rfloor]}^2\Big)^{1/2}.
\end{equation} 
Theorem~\ref{thm:monotone-subgaussian} with $T=(\D-\D)\cap \eps B_2^n$ and $k=\lfloor\delta m/\la\rfloor$ yields that with probability at least $1-2\exp(-c_1 m\delta \log(e\la/\delta)/\la)$, 
\begin{align*}
\sup_{z\in (\D-\D)\cap \eps B_2^n}\Big(\frac{\la}{\delta m} \|Az\|_{[\lfloor\delta m/\la\rfloor]}^2\Big)^{1/2} & \lesssim \sqrt{\frac{\la}{\del m}} \ell_*((\D-\D)\cap \eps B_2^n) + \eps \sqrt{\log(e\la/\delta)},
\end{align*}
which is bounded by $\delta$ if $\eps\lesssim\delta/\sqrt{\log(e\la/\delta)}$ and 
$m\gtrsim\la \del^{-3} \ell_*((\D-\D)\cap \eps B^n_2)^2$.
Combining our estimates we find that if
\begin{align}\label{eqn:GaussianCond}
\la\gtrsim R\sqrt{\log(R/\del)},\quad
m\gtrsim \la^2\del^{-2} \log\mathcal{N}(\D, \eps)+  \la \del^{-3} \ell_*((\D-\D)\cap \eps B^n_2)^2
\end{align}
for $\eps\lesssim \del/\sqrt{\log(e\la/\del)}$, then with probability at least $1-2 \exp(-c\del^2 m/\la^2)$,
\begin{equation*}
\sup_{x,y\in \D}|d_{\lambda}(f(x),f(y))-\|x-y\|_2|\lesssim \del.
\end{equation*}
By rescaling $\del>0$ with a universal constant we obtain the result.
\end{proof}

\section{Proof of Theorem~\ref{thm:Gaussian_Circulant_Dithered}}
\label{sec:proofThmMain2}

Theorem~\ref{thm:Gaussian_Circulant_Dithered} immediately follows from the following result.
\begin{thm}
	\label{thm:Gaussian_Circulant_Dithered_main}
	There exists an absolute constant $c>0$ and a polylogarithmic factor $\alpha$ satisfying $\alpha\leq \log^4(n)+\log(\eta^{-1})$ such that the following holds.
	Let $\D\subset RB^n_2$ for $R\geq 1$, $\eta\in (0,\frac{1}{2}]$, $\del\in (0,1)$, and $\lambda\geq R$. Let $I\subset [n]$ satisfy $|I|=m$, $\Gamma_{\xi}$ be the circulant matrix generated by a mean-zero random vector $\xi\in \R^n$ with independent, variance one and $K$-subgaussian entries, and $D_{\theta}$ is a diagonal matrix containing independent Rademachers on its diagonal. Let $A=R_I\Gamma_{\xi}D_{\theta}$, $\ta,\ta'$ be uniformly distributed in $[-\la,\la]^m$ and $A$, $\tau$, and $\tau'$ be independent. Consider the mappings 
	\begin{equation*}
f(x)=\sign(Ax+\ta), \qquad f'(x)=\sign(Ax+\ta').
\end{equation*}
If
	$$\lambda\gtrsim \alpha R\sqrt{\log(e\la^2/\del R^2)}, \quad R^2\geq \del \lambda^2,$$
	and
	\begin{align*}
		m&\gtrsim \alpha^2\del^{-2} \log(\mathcal{N}(\D,r))+\alpha^2\lambda^{-2}\del^{-3}\ell_*((\D-\D)\cap r B^n_2)^2, 
	\end{align*}
	for a parameter $r\leq c\del R$, then with probability at least $1-\eta$,
	\begin{equation*}\label{eq:binaryembedding_fast_finite}
	\sup_{x,y\in \D}\big|\tfrac{\la^2}{m}\skp{f(x)}{f'(y)} - \skp{x}{y}\big|\leq  \del \la^2.
	\end{equation*}
\end{thm}
\begin{proof}[Proof of Theorem~\ref{thm:Gaussian_Circulant_Dithered}]
	By the definition \eqref{eqn:defbigF} of $F$,  
	\begin{equation*}
	\skp{F(x)}{F(y)}_{S_m}=\skp{f(x)}{f'(y)}+\skp{f(y)}{f'(x)}
	\end{equation*}
	and
	\begin{equation*}
	\skp{F(x)-F(y)}{F(x)-F(y)}_{S_m}=2\skp{f(x)-f(y)}{f'(x)-f'(y)}.
	\end{equation*}
	Hence, by the triangle inequality 
	\begin{align*}
	& \big|\tfrac{\la^2}{2m}\skp{F(x)}{F(y)}_{S_m} - \skp{x}{y}\big|\\
	& \qquad \leq \tfrac{1}{2}\left(\big|\tfrac{\la^2}{m}\skp{f(x)}{f'(y)}-\skp{x}{y}\big|+\big|\tfrac{\la^2}{m}\skp{f(y)}{f'(x)}-\skp{y}{x}\big|\right)
	\end{align*}
	and 
	\begin{align*}
	&\big|\tfrac{\la^2}{2m}\skp{F(x)-F(y)}{F(x)-F(y)}_{S_m}-\|x-y\|_2^2\big|\\
	&\leq \big|\tfrac{\la^2}{m}\skp{f(x)}{f'(x)}  -\|x\|_2^2\big| + \big|\tfrac{\la^2}{m}\skp{f(y)}{f'(y)} - \|y\|_2^2\big| \\
	&\quad + \big|\tfrac{\la^2}{m}\skp{f(x)}{f'(y)}- \skp{x}{y}\big| + \big|\tfrac{\la^2}{m}\skp{f(y)}{f'(x)}- \skp{y}{x}\big|. 
	\end{align*}
	The result now follows from Theorem~\ref{thm:Gaussian_Circulant_Dithered_main}. 
\end{proof}
The remainder of this section is devoted to the proof of Theorem~\ref{thm:Gaussian_Circulant_Dithered_main}. It follows the general ideas of the proof of Theorem~\ref{thm:GaussianDithered}. However, as Theorem~\ref{thm:Gaussian_Circulant_Dithered_main} features a highly structured random embedding matrix with strong stochastic dependencies across the rows, we need to overcome several additional technical hurdles. Our first objective is to prove a version of Theorem~\ref{thm:monotone-subgaussian} for this matrix, which is stated in Corollary~\ref{coro:k-norm_circulant} below. To prove it, we will make use of three ingredients. To state the first one, we use the following terminology from \cite{ORS15}. We say that $A\in \R^{m\times n}$ satisfies RIP$(s,\delta)$ if
$$\Big| \ \|Ax\|_2^2 - \|x\|_2^2\Big|\leq \max\{\delta,\delta^2\} \|x\|_2^2$$
for all $x\in \R^n$ with $\|x\|_0\leq s$. Let $L=\lceil \log_2(n)\rceil $. Given $\delta>0$ and $s\geq 1$, we say that $A\in \R^{m\times n}$ satisfies MRIP$(s,\delta)$ if $A$ satisfies the RIP$(s_{\ell},\delta_{\ell})$ simultaneously for $\ell=1,\ldots,L$, where $s_{\ell} = 2^{\ell} s$ and $\delta_{\ell} = 2^{\ell/2} \delta$. 
\begin{thm}
\label{thm:ORS}
\cite{ORS15} There exist absolute constants $C_1, C_2>0$ such that the following holds.
Let $T\subset\R^n$ and set $\operatorname{rad}(T)=\sup_{x\in T} \|x\|_2$. Let $D_{\theta} \in \R^{n\times n}$ be a diagonal matrix containing i.i.d.\ Rademacher random variables on its diagonal. If $A\in \R^{m\times n}$ satisfies MRIP$(s,\tilde{\delta})$ for 
$$s=C_1(1+\eta), \qquad \tilde{\delta} = C_2 \delta \frac{\operatorname{rad}(T)}{\max\{\operatorname{rad}(T),\ell_*(T)\}}$$
then with probability at least $1-e^{-\eta}$,
$$\sup_{x\in T}\Big| \ \|AD_{\theta} x\|_2^2 - \|x\|_2^2\Big| \leq \max\{\delta,\delta^2\}\operatorname{rad}^2(T).$$
\end{thm}
Our second ingredient can be viewed as a special case of Corollary~\ref{coro:k-norm_circulant} for the set $T=\Sigma_{s,n}:=\{x\in \R^n \; : \; \|x\|_0\leq s, \; \|x\|_2\leq 1\}$ of all $s$-sparse vectors in the Euclidean unit ball. 
\begin{thm}\cite[Lemma 2.1]{DiM18b} \label{thm:structure-3}
Let $\xi\in \R^n$ be a centered, isotropic, $K$-subgaussian random vector.
There exist constants $c_1, c_2>0$ that only depend on $K$ such that the following holds. Let $s\in [n]$. For any $u\geq 1$, with probability at least $1-e^{-c_1u^2}$, 
\begin{equation*}
\sup_{x\in \Sigma_{s,n}}\|\Gamma_\xi x\|_{[s]} \leq c_2 (\sqrt{s}\log(en)\log(s) +  u\sqrt{s}).
\end{equation*}
\end{thm}
The final ingredient is a well-known result on the restricted isometry property of a subsampled random circulant matrix.  
\begin{thm}\cite{KMR14} \label{thm:structure-4}
Let $\xi\in \R^n$ be a centered random vector with independent, variance one, $K$-subgaussian entries. 
There exist constants $c_1, c_2>0$ that only depend on $K$ such that the following holds. Let $I\subset [n]$ with $|I|=m$.
For any $\del>0$, if
\begin{equation*}
m\geq c_1\del^{-2} s \log^2(s)\log^2(n),
\end{equation*}
then the matrix $\tfrac{1}{\sqrt{m}}R_I\Gamma_{\xi}$ satisfies RIP$(s,\delta)$ with probability at least $1-2\exp(-c_2 \del^2 m/s)$. 
\end{thm}

\begin{lemma}\label{lem:new}
Let $A\in \R^{m\times n}$, $\del>0$ and $s\geq 1$. For $l\in \N$ set $s_l=2^l s$, $\del_l=2^{l/2} \del$.
Suppose that $A$ satisfies MRIP$(s,\delta)$ and 
	\begin{align*}
	\sup_{x\in \Sigma_{s_l,n}}\|Ax\|_{[s_l]}\leq \max\{\del_l, \del_l^2\}, \quad \text{for all } l=1, \ldots, \lceil \log_2(n)\rceil.
	\end{align*}
	Then $\begin{bmatrix}
	A  & \operatorname{Id}_{m}
	\end{bmatrix}$ satisfies MRIP$(s, 3\del)$.
\end{lemma}	
\begin{proof}
	Let $l\in  \{1, \ldots, \lceil \log_2(n+m)\rceil \}$. We need to show that $\begin{bmatrix}
	A & \operatorname{Id}_{m}
	\end{bmatrix}$ satisfies RIP$(s_l,(3\del)_l)$. We will derive the stronger estimate 
	\begin{equation}\label{eq:RIP}
	\sup_{z \in \Sigma_{s_l,n+m}} \left|\left\|\begin{bmatrix}
	A & \operatorname{Id}_{m}
	\end{bmatrix}
	z
	\right\|_2^2 - \|z\|_2^2\right|\leq \max\{(3\del)_l,(3\del)_l^2\}.
	\end{equation}
	To show \eqref{eq:RIP}, it suffices to prove that 
	\begin{equation*}
	\sup_{x\in \Sigma_{s_l,n}, y\in \Sigma_{s_l,m}} \left|\left\|\begin{bmatrix}
	A & \operatorname{Id}_{m}
	\end{bmatrix}
	\begin{bmatrix}
	x \\ 
	y
	\end{bmatrix}
	\right\|_2^2 - \left\|\begin{bmatrix}
	x \\ 
	y
	\end{bmatrix}
	\right\|_2^2\right|\leq \max\{(3\del)_l,(3\del)_l^2\}.
	\end{equation*}
	By assumption we have
	\begin{align*}
	&\sup_{x\in \Sigma_{s_l,n}, y\in \Sigma_{s_l,m}}\left|\left\|\begin{bmatrix}
	A & \operatorname{Id}_{m}
	\end{bmatrix}
	\begin{bmatrix}
	x \\ 
	y
	\end{bmatrix}
	\right\|_2^2 - \left\|\begin{bmatrix}
	x \\ 
	y
	\end{bmatrix}
	\right\|_2^2\right| \\
	&\leq \sup_{x\in \Sigma_{s_l,n}} \left|  \|Ax\|_2^2 - \|x\|_2^2\right| + 2 \sup_{x\in \Sigma_{s_l,n}, y\in \Sigma_{s_l,m}}|\langle A x,y\rangle|\\
	&\leq \sup_{x\in \Sigma_{s_l,n}} \left|  \|A x\|_2^2 - \|x\|_2^2\right| + 2 \sup_{x\in \Sigma_{s_l,n}} \|Ax\|_{[s_l]}\\
	&\leq 3 \max\{\del_l,\del_l^2\} \leq \max\{(3\del)_l,(3\del)_l^2\}.
	\end{align*}
\end{proof}

\begin{thm}\label{thm:new}
Let $A\in \R^{m\times n}$, $T\subset \eps B^n_2$, $\del, \eta>0$ and $k\in [m]$. 
There exist absolute constants $C_1, C_2>0$ such that the following holds.
Set $$s=C_1(1+\eta), \qquad \tilde{\delta} = C_2 \del\frac{\operatorname{rad}(\tilde{T})}{\max\{\operatorname{rad}(\tilde{T}),\ell_*(\tilde{T})\}}$$
where $\tilde{T}=(\tfrac{1}{\eps}T) \times \Sigma_{k,m}$. 
Suppose that $A$ satisfies MRIP$(s,\tilde{\delta})$ and 
	\begin{align*}
	\sup_{x\in \Sigma_{s_l,n}}\|Ax\|_{[s_l]}\leq \max\{\tilde{\del}_l, \tilde{\del}_l^2\}, \quad \text{for all } l=1, \ldots, \lceil \log_2(n)\rceil,
	\end{align*}
	where $s_l=2^l s$ and $\tilde{\del}_l=2^{l/2} \tilde{\del}$. Let $\theta\in \{-1,1\}^n$ be a Rademacher random vector. 
Then, with probability at least $1-e^{-\eta}$,
\begin{equation*}
\sup_{x\in T} \Big\| A  D_{\theta}x\Big\|_{[k]}\leq \max\{\del,\del^2\} \eps.
\end{equation*}
\end{thm}	

\begin{proof}
	Let $\zeta\in \{-1,1\}^{m}$ be a Rademacher random vector that is independent of $\theta$. Set $T'=\tfrac{1}{\eps}T$. We have
	\begin{equation}\label{eq:rescaled_set}
	\sup_{x\in T} \Big\| A D_\theta x\Big\|_{[k]}=\eps\cdot \sup_{x\in T'} \Big\| A D_\theta x\Big\|_{[k]}
	\end{equation}
	and
	\begin{align*}
	\sup_{x\in T'} \Big\| A D_\theta x\Big\|_{[k]}  = \sup_{x\in T'} \Big\|D_{\zeta} A D_\theta x\Big\|_{[k]} 
	 = \sup_{x\in T'} \sup_{y\in \Sigma_{k,m}} \Big|\Big\langle D_{\zeta} A D_\theta x, y\Big\rangle\Big|.
	\end{align*}  
	Next,
	\begin{align*}
	 \Big|\Big\langle  A D_\theta x, D_{\zeta}y\Big\rangle\Big| 
	& =  \frac{1}{4}\Big|\Big\|  A D_\theta x+D_{\zeta}y\Big\|_2^2- \Big\| AD_\theta x-D_{\zeta}y\Big\|_2^2\Big| \\
	& = \frac{1}{4}\bigg|\left\|\begin{bmatrix}
	 AD_\theta  & D_{\zeta}
	\end{bmatrix}
	\begin{bmatrix}
	x \\ 
	y
	\end{bmatrix}
	\right\|_2^2 
	-  \left\|\begin{bmatrix}
	x \\ 
	y
	\end{bmatrix}
	\right\|_2^2 \\
	&\quad - \left(\left\|\begin{bmatrix}
	 A D_\theta & D_{\zeta}
	\end{bmatrix}
	\begin{bmatrix}
	x \\ 
	-y
	\end{bmatrix}
	\right\|_2^2 - \left\|\begin{bmatrix}
	x \\ 
	-y
	\end{bmatrix}
	\right\|_2^2\right) \bigg|.
	\end{align*}  
	Clearly,
	$$\begin{bmatrix}
	 A D_{\theta} & D_{\zeta}
	\end{bmatrix} = \begin{bmatrix}
	 A  & \operatorname{Id}_{m}
	\end{bmatrix}
	\begin{bmatrix}
	D_{\theta} & 0 \\
	0 & D_{\zeta}
	\end{bmatrix},
	$$
	where $\operatorname{Id}_{m}\in \R^{m\times m}$ is the identity matrix.
	Combining the calculations above, 
	\begin{align}\label{eq:setting_of_thm:ORS}
	\sup_{x\in T'} \Big\| A D_{\theta}x\Big\|_{[k]}
	\leq \frac{1}{2}\sup_{(x,y)^T\in \tilde{T}}\left|\left\|\begin{bmatrix}
	 A & \operatorname{Id}_{m}
	\end{bmatrix}
	\begin{bmatrix}
	D_{\theta} & 0 \\
	0 & D_{\zeta}
	\end{bmatrix}
	\begin{bmatrix}
	x \\ 
	y
	\end{bmatrix}
	\right\|_2^2 
	-  \left\|\begin{bmatrix}
	x \\ 
	y
	\end{bmatrix}
	\right\|_2^2\right|,
	\end{align}
	where $\tilde{T}=T'\times \Sigma_{k,m}$. 
	Since Lemma~\ref{lem:new} shows that 
	$\begin{bmatrix}
	A  & \operatorname{Id}_{m}
	\end{bmatrix}$ satisfies MRIP$(s, 3\tilde{\del})$, Theorem~\ref{thm:ORS} in combination with \eqref{eq:setting_of_thm:ORS} and \eqref{eq:rescaled_set}  
	implies that with probability at least
	$1-e^{-\eta}$,
	\begin{equation*}
	\sup_{x\in T} \Big\| A D_{\theta}x\Big\|_{[k]}\leq \frac{1}{2}\eps \max\{3\del,(3\del)^2\} \rad^2(\tilde{T}).
	\end{equation*}
	Using $\operatorname{rad}(\tilde{T})\leq \sqrt{2}$ and rescaling $C_2$ by an absolute constant factor yields the result.
\end{proof}

\begin{thm}\label{thm:k-norm_circulant}
Let $\xi\in \R^n$ be a centered random vector with independent, variance one, $K$-subgaussian entries. 
Let $\theta\in \{-1,1\}^n$ be an independent Rademacher random vector.
There exist absolute constants $c_1, c_2>0$ that only depend on $K$ such that the following holds.
Let $T\subset \eps B^n_2$ and $I\subset [n]$ with $|I|=m$. For every $k\leq m$ and $u>0$, if 
\begin{equation}
m\geq c_1 \max\{u^{-1}, u^{-2}\}(\eps^{-2}\ell_*(T)^2+ k\log(em/k))\log^8(n),
\end{equation} 
then 
\begin{equation*}
\sup_{x\in T} \Big\|\tfrac{1}{\sqrt{m}}R_I\Gamma_\xi D_{\theta}x\Big\|_{[k]}\leq  u \eps
\end{equation*}
with probability at least $$1-2\exp(-c_2 \min\{\sqrt{u}, u\} \sqrt{m}/(\eps^{-1}\ell_*(T) +\sqrt{k\log(em/k)})).$$
\end{thm}
\begin{proof} Let $s\geq 1$, $\del>0$ and set $A=\tfrac{1}{\sqrt{m}}R_I\Gamma_{\xi}$.
From Theorem~\ref{thm:structure-4} and Theorem~\ref{thm:structure-3} it follows that if 
\begin{equation}\label{eq:cond_1_proof_thm:k-norm}
m\geq c_1 \del^{-2}s\log^2(s)\log^2(n),
\end{equation}	
then with probability at least $1-2\exp(-c_2\del^2 m/s)$ the following holds: 
$A$ satisfies RIP$(s,\del)$ and 
\begin{equation}\label{eq:cond_2_proof_thm:k-norm}
\sup_{x\in \Sigma_{s,n}}\|Ax\|_{[s]}\leq \del.
\end{equation}
A union bound shows that if 
\begin{equation}
m\geq c_1 \del^{-2}s\log^4(n),
\end{equation}	
then, with probability at least $1-2\exp(-c_2\del^2 m/s)$, the matrix $A$ satisfies MRIP$(s,\del)$ and \eqref{eq:cond_2_proof_thm:k-norm}
holds for $s=s_l=2^l s$ and $\del=\del_l=2^{l/2}\del$ simultaneously for all $l=1, \ldots, \lceil \log_2(n)\rceil$.
Next, fix 
$$s=C_1(1+\eta), \qquad \tilde{\delta} = C_2 \del\frac{\operatorname{rad}(\tilde{T})}{\max\{\operatorname{rad}(\tilde{T}),\ell_*(\tilde{T})\}}$$
where $\tilde{T}=T'\times \Sigma_{k,m}$, $T'=\tfrac{1}{\eps}T$ and $C_1, C_2>0$ are absolute constants from 
Theorem~\ref{thm:new}. 
Clearly,
$$\ell_*(\tilde{T}) \simeq \ell_*(T')+\ell_*(\Sigma_{k,m})\simeq \ell_*(T') + \sqrt{k\log(em/k)}$$
and 
$$\operatorname{rad}(\tilde{T}) \simeq \operatorname{rad}(T')+1\simeq 1,$$
which implies $$\tilde{\del}\simeq \frac{\del}{\ell_*(T') +\sqrt{k\log(em/k)}}.$$
Consequently, if 
\begin{equation}\label{eq:cond_|I|}
m\gtrsim \del^{-2} (\ell_*(T')^2 + k\log(em/k)) (1+\eta) \log^4(n),
\end{equation} 
then the assumptions of Theorem~\ref{thm:new} are satisfied for $A=\tfrac{1}{\sqrt{m}}R_I\Gamma_{\xi}$ with probability at least 
$1-2\exp(-c_2\tilde{\del}^2 m/s)$. Hence, if \eqref{eq:cond_|I|} holds then 
\begin{equation*}
\sup_{x\in T} \Big\| A  D_{\theta}x\Big\|_{[k]}\leq \max\{\del,\del^2\}  \eps
\end{equation*}
with probability at least $1-2\exp(-c_2\tilde{\del}^2 m/\eta)-\exp(-\eta)$.
Choosing $\eta=\tilde{\delta}\sqrt{m}$ we conclude that if 
\begin{equation*}
m\gtrsim \del^{-2}(\ell_*(T')^2+ k\log(em/k))\log^8(n)
\end{equation*}
then 
\begin{equation*}
\sup_{x\in T} \Big\|A D_{\theta}x\Big\|_{[k]}\leq \max\{\del,\del^2\} \eps
\end{equation*}
with probability at least $1-2\exp(-c_7 \del \sqrt{m}/(\ell_*(T') +\sqrt{k\log(em/k)}))$.
Since $$\max\{\min\{\sqrt{u},u\}, \min\{u,u^2\}\}=u,$$ the result now follows by substituting $\del$ by $\min\{\sqrt{u}, u\}$.
\end{proof}
Applying Theorem~\ref{thm:k-norm_circulant} with $u=\tfrac{C(\eps^{-1}\ell_*(T)+\sqrt{k\log(em/k)})u'\log^4(n)}{\sqrt{m}}$ where 
$C>0$ denotes an absolute constant and $u'\geq 1$ yields the following corollary: 
\begin{coro}\label{coro:k-norm_circulant}
	Let $\xi\in \R^n$ be a centered random vector with independent, variance one, $K$-subgaussian entries. 
	Let $\theta\in \{-1,1\}^n$ be an independent Rademacher random vector.
	There exist absolute constants $c_1, c_2, c_3>0$ that only depend on $K$ such that the following holds.
	Let $T\subset \eps B^n_2$ and $I\subset [n]$ with $|I|=m$. For every $k\leq m$ and $u'\geq 1$, if
	\begin{equation*}
	m\geq c_1(\eps^{-2}\ell_*(T)^2+ k\log(em/k))(u')^2\log^8(n),
	\end{equation*}
	then 
	\begin{equation*}
	\sup_{x\in T} \Big\|R_I\Gamma_\xi D_{\theta}x\Big\|_{[k]}\leq c_2(\ell_*(T)+\eps \sqrt{k\log(em/k)})u' \log^4(n)
	\end{equation*}
	with probability at least $1-2\exp(-c_3 u' \log^4(n))$.
\end{coro}	

The second ingredient for the proof of Theorem~\ref{thm:Gaussian_Circulant_Dithered_main} is a consequence of Theorem~\ref{thm:k-norm_circulant}.

\begin{lemma}\label{lem:bias_terms}
Let $\xi\in \R^n$ be a centered random vector with independent, variance one, $K$-subgaussian entries. 
Let $\theta\in \{-1,1\}^n$ be an independent Rademacher random vector.
Set $A=R_I\Gamma_\xi D_{\theta}$ for $I\subset[n]$ with $|I|=m$.
There exist absolute constants $c_1, c_2, c_3>0$ that only depend on $K$ such that the following holds.
For any $T\subset RB^n_2$, $\la>0$ and $\del\in (0, 1)$,
if 
\begin{equation*}
m\geq c_1\delta^{-1}R^{-2}\ell_*(T)^2 \log^{8}(n), \quad  \la\geq c_2 R\sqrt{\log(e\la^2/\del R^2)}\log^4(n),
\end{equation*}
then 
\begin{equation*}
\sup_{x,y\in T} \frac{1}{m}\sum_{i=1}^m |\skp{a_i}{x}||\skp{a_i}{y}|1_{|\skp{a_i}{x}|> \la}\leq \del R^2
\end{equation*}
with probability at least $1-2\exp(-c_3 R\sqrt{\del m}/\ell_*(T))-2\exp(-c_3 \lambda/R\sqrt{\log(e\la^2/\del R^2)})$.
\end{lemma}
\begin{proof}
For $k\in [m]$, define the event 
\begin{equation*}
\mathcal{A}_k=\Big\{\sup_{x\in T}  \sum_{i=1}^m 1_{|\skp{a_i}{x}|> \la}\leq k\Big\} .
\end{equation*}
On the event $\mathcal{A}_k$,
\begin{align*}
&\sup_{x,y\in T} \frac{1}{m}\sum_{i=1}^m |\skp{a_i}{x}||\skp{a_i}{y}|1_{|\skp{a_i}{x}|> \la}\\
&\qquad \leq \sup_{x,y\in T} \max_{S\subset [m], |S|\leq k}\frac{1}{m}\sum_{i\in S}|\skp{a_i}{x}||\skp{a_i}{y}|\\
&\qquad \leq \sup_{x,y\in T} \max_{S\subset [m], |S|\leq k}\frac{1}{m}\Big(\sum_{i\in S}|\skp{a_i}{x}|^2\Big)^{1/2} \Big(\sum_{i\in S}|\skp{a_i}{y}|^2\Big)^{1/2}\\
&\qquad \leq \Big(\sup_{x\in T} \Big\|\tfrac{1}{\sqrt{m}}Ax\Big\|_{[k]}\Big)^{2}.
\end{align*}
Analogously to \eqref{eqn:largeCoeffk-norm1} and \eqref{eqn:largeCoeffk-norm2}, we see that
\begin{align*}
\Big\{\sup_{x\in T} \tfrac{1}{\sqrt{m}}\|Ax\|_{[k]}\leq \la\sqrt{\tfrac{k}{m}}  \Big\} \subset \mathcal{A}_k.
\end{align*}
Consequently, for every $\del>0$ and every $k\in [m]$,
\begin{align*}
\bP\Big(\sup_{x,y\in T} \frac{1}{m}\sum_{i=1}^m |\skp{a_i}{x}||\skp{a_i}{y}|1_{|\skp{a_i}{x}|> \la}>\del R^2\Big)
&\leq \bP\Big(\sup_{x\in T} \|\tfrac{1}{\sqrt{m}}Ax\|_{[k]}>\la\sqrt{\tfrac{k}{m}}\Big)\\
&\quad + \bP\Big(\sup_{x\in T} \|\tfrac{1}{\sqrt{m}}Ax\|_{[k]}>\sqrt{\del}R\Big).
\end{align*}
For $k=m\del R^2/\la^2\in [m]$ we obtain
\begin{align*}
&\bP\Big(\sup_{x,y\in T} \frac{1}{m}\sum_{i=1}^m |\skp{a_i}{x}||\skp{a_i}{y}|1_{|\skp{a_i}{x}|> \la}>\del R^2\Big)\\
&\leq 2\bP\Big(\sup_{x\in T} \|\tfrac{1}{\sqrt{m}}Ax\|_{[m\del R^2/\la^2]}>\sqrt{\del} R\Big).
\end{align*}
Applying Theorem~\ref{thm:k-norm_circulant} with $k=m\del R^2/\la^2, \eps=R$ and $u=\sqrt{\del}$ yields the result. 
\end{proof}

The final ingredient for the proof of Theorem~\ref{thm:Gaussian_Circulant_Dithered_main} is the following observation, which is analogous to Lemma~\ref{lem:exp}. 	
\begin{lemma}
\label{lem:exp_mult} 
Let $\lambda>0$ and let $\tae, \tae'$ be independent and uniformly distributed in $[-\la,\la]$. Define the function 
\begin{equation*}
\phi_\la(s)=(|s|-\la)1_{|s|> \la} \qquad \text{for } s\in \R.
\end{equation*}	
Then for $a,b\in \R$,
\begin{align*}
|\la^2\E[\sgn(a+\tae)\sgn(b+\tae')] - ab| \leq \phi_\la(a)\,|b| + |a|\,\phi_\la(b). 
\end{align*}
\end{lemma}
\begin{proof}
By independence of $\tae$ and $\tae'$,
\begin{equation}\label{eq:lem:exp_mult_independence}
\E[\sgn(a+\tae)\sgn(b+\tae')]=\E[\sgn(a+\tae)]\E[\sgn(b+\tae')].
\end{equation}
If $|a|\leq \lambda$, then
\begin{align}\label{eq:lem:exp_mult_sign}
\E[\sgn(a+\tae)]&=\bP(a+\tae\geq 0) -  \bP(a+\tae< 0)\\ \nonumber
&=\bP(\tae\geq -a) -  \bP(\tae< -a)\\ \nonumber
&=\frac{\la+a}{2\la} - \frac{\la-a}{2\la}\\ \nonumber
&=\frac{a}{\la}.
\end{align}
Moreover, if $|a|>\la$, then $\E[\sgn(a+\tae)]=\sgn(a)$.
We distinguish three cases:
\begin{enumerate}
\item[1)] $|a|\leq \la$ and $|b|\leq \la$,
\item[2)] $|a|\leq \la$ and $|b|>\la$ or $|a|> \la$ and $|b|\leq\la$,
\item[3)] $|a|> \la$ and $|b|> \la$.
\end{enumerate}
In the first case, we have $\phi_\la(a)=\phi_\la(b)=0$. Hence, we need to show that $$\E[\sgn(a+\tae)\sgn(b+\tae')] = \frac{ab}{\la^2}.$$
This follows from \eqref{eq:lem:exp_mult_independence} and \eqref{eq:lem:exp_mult_sign}. In the second case, if $|a|\leq \la$ and $|b|>\la$,
then $\E[\sgn(a+\tae)\sgn(b+\tae')]=\frac{a}{\la}\sign(b)$. Therefore, 
\begin{align*}
|\la^2\E[\sgn(a+\tae)\sgn(b+\tae')]-ab|=|a\la\sign(b)-ab|&=|a||\la \sgn(b)-|b|\sgn(b)|\\
&=|a|(|b|-\la)\\
&=|a|\phi_\la(b).
\end{align*}
Analogously we obtain $|\la^2\E[\sgn(a+\tae)\sgn(b+\tae')]-ab|\leq |b|\phi_\la(a)$ if $|a|>\la$ and $|b|\leq \la$.
In the third case, we have
\begin{align*}
|\la^2\E[\sgn(a+\tae)\sgn(b+\tae')]-ab|&=|\la \sign(a) \la \sign(b)-|a| \sign(a) |b| \sign(b)|\\
&\leq |\la \sign(a) \la \sign(b) - \la \sign(a)  |b| \sign(b)|\\
&\quad + |\la \sign(a)  |b| \sign(b) - |a| \sign(a) |b| \sign(b)|\\
&=\la(|b|-\la)  + (|a|-\la)|b|\\
&\leq |a|(|b|-\la)  + (|a|-\la)|b|.
\end{align*}
This shows the result.
\end{proof}

We are now ready to prove the main result of this section.\\ 

\begin{proof}[Proof of Theorem~\ref{thm:Gaussian_Circulant_Dithered_main}]
For $r\leq  \del R$ let $N_r \subset \D$ be a minimal $r$-net of $\D$ with respect to the Euclidean metric. For every $x\in \D$ we fix a $\pi(x)$ in $\text{argmin}_{z\in N_r}\|x-z\|_2$.
By the triangle inequality,
\begin{align*}
\big|\tfrac{\la^2}{m}\skp{f(x)}{f'(y)}- \skp{x}{y}\big|&\leq \big|\tfrac{\la^2}{m}\skp{f(x)}{f'(y)}- \tfrac{\la^2}{m}\skp{f(\pi(x))}{f'(\pi(y))}\big|\\ 
&\quad + \big|\tfrac{\la^2}{m}\skp{f(\pi(x))}{f'(\pi(y))} - \E_{\tau, \tau'}\tfrac{\la^2}{m}\skp{f(\pi(x))}{f'(\pi(y))}\big|\\
&\quad + \big|\E_{\tau, \tau'}\tfrac{\la^2}{m}\skp{f(\pi(x))}{f'(\pi(y))} - \skp{\pi(x)}{\pi(y)}\big|\\
&\quad + \big|\skp{\pi(x)}{\pi(y)}- \skp{x}{y}\big|=: (1)+ (2) + (3) + (4).
\end{align*}
Next, we bound each of the four summands above uniformly for all $x, y\in \D$.
Clearly, $(4)\leq 2R r\leq 2\del R^2$. \par \vskip 0.1cm
\textbf{Estimate of $(2)$.}
Hoeffding's inequality and the union bound imply that if $m\gtrsim \del^{-2}\log(\mathcal{N}(\D,r)/\eta)$ then 
with probability at least $1-\eta$, 
\begin{equation*}
\sup_{x,y\in \D}\big|\tfrac{\la^2}{m}\skp{f(\pi(x))}{f'(\pi(y))} - \E_{\tau, \tau'}\tfrac{\la^2}{m}\skp{f(\pi(x))}{f'(\pi(y))}\big|\leq \del \la^2.
\end{equation*}
\indent \textbf{Estimate of $(1)$.}
For all $x, y\in \D$, it follows analogously to \eqref{eq:applyDM18} that
\begin{align*}
&\big|\tfrac{\la^2}{m}\skp{f(x)}{f'(y)}- \tfrac{\la^2}{m}\skp{f(\pi(x))}{f'(\pi(y))}\big|\\ \nonumber
&\qquad \leq\tfrac{\la^2}{m}\sum_{i=1}^m |\sign(\skp{a_i}{x}+\tau_i)\sign(\skp{a_i}{y}+\tau_i') \\
&\qquad  \qquad \qquad\qquad\qquad
 -\sign(\skp{a_i}{\pi(x)}+\tau_i)\sign(\skp{a_i}{\pi(y)}+\tau_i')|\\ \nonumber
&\qquad \leq \tfrac{2\la^2}{m}\sum_{i=1}^m 1_{\sign(\skp{a_i}{x}+\tau_i)\neq \sign(\skp{a_i}{\pi(x)}+\tau_i)}+
 1_{\sign(\skp{a_i}{y}+\tau_i')\neq \sign(\skp{a_i}{\pi(y)}+\tau_i')}.\\
\end{align*}
Therefore,
\begin{align*}
&\sup_{x,y\in \D}\big|\tfrac{\la^2}{m}\skp{f(x)}{f'(y)}- \tfrac{\la^2}{m}\skp{f(\pi(x))}{f'(\pi(y))}\big|\\
&\qquad \leq 2\sup_{x\in \D}\tfrac{\la^2}{m}\sum_{i=1}^m  1_{\sign(\skp{a_i}{x}+\tau_i)\neq \sign(\skp{a_i}{\pi(x)}+\tau_i)}\\ 
&\qquad \qquad \quad + 2\sup_{x\in \D}\tfrac{\la^2}{m}\sum_{i=1}^m  1_{\sign(\skp{a_i}{x}+\tau_i')\neq \sign(\skp{a_i}{\pi(x)}+\tau_i')}.
\end{align*}
By \eqref{eqn:signChangeCov}
\begin{align*}
1_{\sign(\skp{a_i}{x}+\tau_i)\neq \sign(\skp{a_i}{\pi(x)}+\tau_i)}
\leq 1_{|\langle a_i,\pi(x)\rangle+\tau_i|\leq \delta\la} + 1_{|\langle a_i,\pi(x)-x\rangle|> \delta\la}
\end{align*}
and hence
\begin{align*}
&\sup_{x,y\in \D}\big|\tfrac{\la^2}{m}\skp{f(x)}{f'(y)}- \tfrac{\la^2}{m}\skp{f(\pi(x))}{f'(\pi(y))}\big|\\
&\qquad \leq 2\sup_{x\in N_{r}} \tfrac{\la^2}{m}\sum_{i=1}^m  1_{|\langle a_i,x\rangle+\tau_i|\leq \delta\la} 
+ 2\sup_{x\in N_{r}} \tfrac{\la^2}{m}\sum_{i=1}^m 1_{|\langle a_i,x\rangle+\tau_i'|\leq \delta\la} \\
&\qquad \quad+ 4\sup_{z\in (\D-\D)\cap r B_2^n} \tfrac{\la^2}{m}\sum_{i=1}^m 1_{|\langle a_i,z\rangle|>\delta\la}.
\end{align*}
For every $x\in \R^n$,
$$\bP_{\tau}(|\langle a_i,x\rangle+\tau_i|\leq \delta\la)\leq \del.$$
Therefore, the Chernoff bound implies that 
$$\sum_{i=1}^m 1_{|\langle a_i,x\rangle+\tau_i|\leq \delta\la}\leq 2\delta m$$
with probability at least $1-\exp(-cm\delta)$. By the union bound, if $m\gtrsim \del^{-1} \log|N_r|$
then with probability at least $1-\exp(-c'm\delta)$, 
$$\sup_{x\in N_{r}} \frac{\la^2}{m}\sum_{i=1}^m  1_{|\langle a_i,x\rangle+\tau_i|\leq \delta\la} \lesssim \delta\la^2.$$
Observe that 
\begin{align*}
\sup_{z\in (\D-\D)\cap r B_2^n} \frac{1}{m}\sum_{i=1}^m 1_{|\langle a_i,z\rangle|>\delta\la}\leq \del
\end{align*}
if 
\begin{align}\label{eq:uniform_k-norm_bound}
\sup_{z\in (\D-\D)\cap r B_2^n} \|Az \|_{[\del m]}\leq \del\la \sqrt{\del m}.
\end{align}
Since $T=(\D-\D)\cap r B^n_2\subset \del R B^n_2$ and $\sqrt{\del}\lambda/R\leq 1$,
applying Theorem~\ref{thm:k-norm_circulant} with $k=\del m$, $T=(\D-\D)\cap r B^n_2$, $\eps=\del R$, $u=\sqrt{\del}\lambda/R$
shows that \eqref{eq:uniform_k-norm_bound} holds with probability at least $1-\eta$, if
$$\lambda\gtrsim R\sqrt{\log(e/\del)}(\log^4(n)+\log(\eta^{-1}))$$ and 
\begin{equation*}\label{eq:cond_I}
m\gtrsim \lambda^{-2}\del^{-3}\ell_*((\D-\D)\cap r B^n_2)^2\big(\log^8(n)+\log^2(\eta^{-1})\big).
\end{equation*} 
\par \vskip 0.1cm
\textbf{Estimate of $(3)$.}
By the triangle inequality, 
\begin{align}\label{eq:bias_polar_summands}
& |\E_{\tau, \tau'}\tfrac{\la^2}{m}\skp{f(\pi(x))}{f'(\pi(y))} - \skp{\pi(x)}{\pi(y)}|\\ \nonumber
&\leq   |\E_{\tau, \tau'}\tfrac{\la^2}{m}\skp{f(\pi(x))}{f'(\pi(y))}-\tfrac{1}{m}\skp{A\pi(x)}{A\pi(y)}|\\ \nonumber
&\quad + |\tfrac{1}{m}\skp{A\pi(x)}{A\pi(y)}-\skp{\pi(x)}{\pi(y)}|.
\end{align}
To bound the first summand, we apply Lemma~\ref{lem:exp_mult} which yields
\begin{align*}
&|\E_{\tau, \tau'}\tfrac{\la^2}{m}\skp{f(\pi(x))}{f'(\pi(y))} - \tfrac{1}{m}\skp{A\pi(x)}{A\pi(y)}|\\
&\leq \frac{1}{m}\sum_{i=1}^m  \E_{\tau, \tau'}|\la^2\sign(\skp{a_i}{\pi(x)} + \ta_i)\sign(\skp{a_i}{\pi(y)}+\ta_i') - \skp{a_i}{\pi(x)}\skp{a_i}{\pi(y)}|\\
&\leq \frac{1}{m}\sum_{i=1}^m \phi_\la(\skp{a_i}{\pi(x)})|\skp{a_i}{\pi(y)}| + \phi_\la(\skp{a_i}{\pi(y)})|\skp{a_i}{\pi(x)}|.\\
\end{align*}
Hence, 
\begin{align*}
&\sup_{x,y\in \D}|\E_{\tau, \tau'}\tfrac{\la^2}{m}\skp{f(\pi(x))}{f'(\pi(y))} - \tfrac{1}{m}\skp{A\pi(x)}{A\pi(y)}|\\
&\qquad \leq 2 \sup_{x, y\in N_r}\frac{1}{m}\sum_{i=1}^m |\skp{a_i}{x}||\skp{a_i}{y}|1_{|\skp{a_i}{x}|> \la}.
\end{align*}
From Lemma~\ref{lem:bias_terms} and $\ell_*(N_r)\lesssim R\sqrt{\log(|N_r|)}$ it follows that if 
\begin{align*}\label{eq:meas_lem:bias_terms}
m&\gtrsim \delta^{-1}\log(|N_r|)(\log^8(n)+\log^2(\eta^{-1})), \\
\la&\gtrsim R\sqrt{\log(e\lambda^2/\del R^2)}(\log^{4}(n)+\log(\eta^{-1})),
\end{align*} 
then
\begin{equation*}
\sup_{x, y\in N_r}\frac{1}{m}\sum_{i=1}^m |\skp{a_i}{x}||\skp{a_i}{y}|1_{|\skp{a_i}{x}|> \la}\leq \del R^2
\end{equation*}
with probability at least $1-\eta$.
Combining the polarization identity with the triangle inequality we estimate the second summand on the right hand side of \eqref{eq:bias_polar_summands} by
\begin{align*}
\big|\tfrac{1}{m}\skp{A\pi(x)}{A\pi(y)} - \skp{\pi(x)}{\pi(y)}\big|
&\leq \tfrac{1}{4}\big|\|\tfrac{1}{\sqrt{m}}A(\pi(x)+\pi(y))\|_2^2-\|\pi(x)+\pi(y)\|_2^2\big| \\
&\quad + \tfrac{1}{4}\big|\|\tfrac{1}{\sqrt{m}}A(\pi(x)-\pi(y))\|_2^2-\|\pi(x)-\pi(y)\|_2^2\big|.
\end{align*}
Theorem~\ref{thm:structure-4} in combination with \cite[Theorem~3.1]{KrW11} yields the following:
there exist absolute constants $c_1, c_2>0$ such that if
\begin{equation*}
m\geq c_1 \del^{-2} \log(|N_r|/\eta)\big(\log^4(n)+\log(\eta^{-1})\big),
\end{equation*}
then with probability at least $1-\eta$,
\begin{equation*}
\Big|\|\tfrac{1}{\sqrt{m}}A(\pi(x)+\pi(y))\|_2^2-\|\pi(x)+\pi(y)\|_2^2\Big|\leq \del \|\pi(x)+\pi(y)\|_2^2 \quad \text{ for all } x,y\in \D
\end{equation*}
and 
\begin{equation*}
\Big|\|\tfrac{1}{\sqrt{m}}A(\pi(x)-\pi(y))\|_2^2-\|\pi(x)-\pi(y)\|_2^2\Big|\leq \del \|\pi(x)-\pi(y)\|_2^2 \quad \text{ for all } x,y\in \D.
\end{equation*}
Combining our estimates, we find that there exist absolute constants $C, C'\geq 1$ such that if
$$\lambda\gtrsim R\sqrt{\log(e\lambda^2/\del R^2)}(\log^{4}(n)+\log(\eta^{-1})),\quad R^2\geq \del \lambda^2$$
and
\begin{align*}
m&\gtrsim \lambda^{-2}\del^{-3}\ell_*((\D-\D)\cap r B^n_2)^2\big(\log^8(n)+\log^2(\eta^{-1})\big), \\
m&\gtrsim \delta^{-1}\log(\mathcal{N}(\D,r))(\log^8(n)+\log^2(\eta^{-1})), \\
m&\gtrsim  \del^{-2} \log(\mathcal{N}(\D,r)/\eta)\big(\log^4(n)+\log(\eta^{-1})\big),
\end{align*}
for $r\leq \del R$,
then, with probability at least $1-C\eta$,
\begin{equation*}
\sup_{x,y\in \D}|\tfrac{\la^2}{m}\skp{f(x)}{f'(y)}-\skp{x}{y}|\leq C' \del\la^2.
\end{equation*}
By rescaling $\del$ and $\eta$ we obtain the result.
\end{proof}

\section{Lower bounds}

In this section we derive bounds on the necessary number of bits required by an oblivious random embedding that preserves Euclidean or squared Euclidean distances.  
\begin{thm}
\label{thm:optimalbit_Euclidean}
	Let $N\in \N$, $\eta\in (0,\tfrac{1}{2})$, $R>0$, $\del\in (0,\tfrac{R}{2})$ 
	and assume that 
	\begin{equation}
	\label{eqn:optimalbit_Euclidean}
	n\gtrsim \del^{-2} R^2\log(N/\eta),\qquad \sqrt{N/\eta}\gtrsim \del^{-2}R^2\log(N/\eta).
	\end{equation}
	Let $f: R B^n_2\to \{-1,1\}^m$ be a random binary embedding map and $d:\{-1,1\}^m\times \{-1,1\}^m\to \R$ a (random) reconstruction map such that given any finite data set $\D\subset RB^n_2$ with $|\D|=N$,
	\begin{equation*}
	\bP\Big(\sup_{x,y\in \D}|d(f(x),f(y))- \eu{x-y}|\leq \del\Big)\geq 1-\eta.
	\end{equation*}
	Then $m\gtrsim \del^{-2} R^2\log(N/\eta)$.
\end{thm}
By considering $R^{-1}\D\subset B^n_2$ it suffices to prove the result for $R=1$. Further, if 
$|d(f(x),f(y))- \eu{x-y}|\leq 1$ and $x, y\in B^n_2$ then
\begin{align*}
|d(f(x),f(y))^2- \eu{x-y}^2|&\leq |d(f(x),f(y))+ \eu{x-y}|\cdot |d(f(x),f(y))- \eu{x-y}|\\
&\leq 5 \cdot |d(f(x),f(y))- \eu{x-y}|.
\end{align*}
Therefore, Theorem~\ref{thm:optimalbit_Euclidean} with $R=1$ follows from the following result:
\begin{thm}\label{thm:optimalbit_squaredEuclidean}
	Let $N\in \N, \eta, \del\in (0,\tfrac{1}{2})$ and assume that \eqref{eqn:optimalbit_Euclidean} holds for $R=1$.
	Let $f: B^n_2\to \{-1,1\}^m$ be a random binary embedding map and $d:\{-1,1\}^m\times \{-1,1\}^m\to \R$ a (random) reconstruction map such that given any finite data set $\D\subset B^n_2$ with $|\D|=N$,
	\begin{equation}\label{eq:preservation_assumption}
	\bP\Big(\sup_{x,y\in \D}|d(f(x),f(y))- \eu{x-y}^2|\leq \del\Big)\geq 1-\eta.
	\end{equation}
	Then $m\gtrsim \del^{-2} \log(N/\eta)$.
\end{thm}
Clearly this result is equivalent to the following statement.
\begin{thm}\label{thm:optimalbit_inner}
	Let $N\in \N, \eta, \del\in (0,\tfrac{1}{2})$ and assume that \eqref{eqn:optimalbit_Euclidean} holds for $R=1$.
	Let $f: B^n_2\to \{-1,1\}^m$ be a random binary embedding map and $d:\{-1,1\}^m\times \{-1,1\}^m\to \R$ a (random) reconstruction map such that given any finite data set $\D\subset B^n_2$ with $|\D|=N$,
	\begin{equation}\label{eq:preservation_assumption}
	\bP\Big(\sup_{x,y\in \D}|d(f(x),f(y))- \skp{x}{y}|\leq \del\Big)\geq 1-\eta.
	\end{equation}
	Then $m\gtrsim \del^{-2} \log(N/\eta)$.
\end{thm}

\begin{lemma}\label{lem:sphere:uniformmeasure} 
	Let $V$ be uniformly distributed on the sphere $\Sp^{k-1}$. Let $\del\in (0,1]$ and $k\geq 2\del^{-2}$. Then for any $x\in \Sp^{k-1}$,
	\begin{equation}\label{eq:sphere:bound}
	\tfrac{1}{6\del\sqrt{k}}(1-\del^2)^{\tfrac{k-1}{2}}\leq \bP(\skp{V}{x}\geq \del)\leq \tfrac{1}{2\del\sqrt{k}}(1-\del^2)^{\tfrac{k-1}{2}}.
	\end{equation} 
	If additionally $\del\leq\tfrac{1}{\sqrt{2}}$, then
	\begin{equation}\label{eq:sphere:lowerbound}
	\bP(|\skp{V}{x}|\geq \del)\geq \exp(-2\del^2k).
	\end{equation}
	Moreover, for $\alpha\in (0,1]$ and $k\geq 8$,
	\begin{equation}\label{eq:sphere:small_ball}
	\bP(\eu{V-x}\leq \alpha)\leq \tfrac{1}{\sqrt{k}}\alpha^{k-1}.
	\end{equation} 
\end{lemma}	
\begin{proof}
Inequality \eqref{eq:sphere:bound} is standard (e.g. see \cite[Section~7.2]{boucheron2013concentration}). The second inequality immediately follows from the lower bound in \eqref{eq:sphere:bound} and using $1-u\geq \exp(-2u)$ for $u\in [0, \tfrac{1}{2}]$. Inequality \eqref{eq:sphere:small_ball} follows from the upper bound in \eqref{eq:sphere:bound}
for $\del = 1-\tfrac{\alpha^2}{2}$ by observing that $\eu{V-x}\leq \alpha$ if and only if $\skp{V}{x}\geq 1-\tfrac{\alpha^2}{2}$.
\end{proof}
\begin{lemma}\label{lem:packing:sphere} For $\del\in (0, 1]$ let $\mathcal{N}\subset \Sp^{k-1}$ be a maximal $\del$-packing of $\Sp^{k-1}$ (a maximal subset of $\Sp^{k-1}$ such that 
$\eu{x-x'}>\del$ for all distinct $x, x'\in \mathcal{N}$). Then
\begin{equation}
\sqrt{k}\cdot\big(\tfrac{1}{\del}\big)^{k-1}\leq|\mathcal{N}|\leq \big(\tfrac{4}{\del}+1\big)^k.
\end{equation}
\end{lemma}
\begin{proof}
The upper bound follows from \cite[Lemma 4.2.8 and Corollary 4.2.13]{RomanHDP}. Define
$\mathcal{B}_\del(x)=\{y\in \Sp^{k-1}\; : \; \eu{x-y}\leq \del\}$.
Since $\Sp^{k-1}$ is covered by the union of balls $\mathcal{B}_\del(x)$ with centers $x\in \mathcal{N}$, it follows 
$$1\leq \sum_{x\in \mathcal{N}} \bP(\eu{V-x}\leq \del),$$
where $V$ is uniformly distributed on $\Sp^{k-1}$. The claim now follows from inequality \eqref{eq:sphere:small_ball}.
\end{proof}	
Theorem~\ref{thm:optimalbit_inner} is clearly implied by the following result:
\begin{thm}\label{thm:optimalbit_inner_2}
	There exists an absolute constant $\alpha\in (0,\tfrac{1}{2})$ such that the following holds.
	Let $N\in \N, \del, \eta\in (0,\tfrac{1}{2})$ and assume that 
	\begin{equation}
	n\gtrsim \del^{-2} \log(N/\eta), \qquad \sqrt{N/\eta}\gtrsim \del^{-2}\log(N/\eta).
	\end{equation}
	Let $f: B^n_2\to \{-1,1\}^m$ be a random binary embedding map and $d:\{-1,1\}^m\times \{-1,1\}^m\to \R$ a random reconstruction map such that given any finite data set $\D\subset B^n_2$ with $|\D|=N$ 
	and $(\D-\D)\cap \alpha B^n_2=\{0\}$, 
	\begin{equation}\label{eq:optimalbit_inner_2:assumption}
	\bP\Big(\sup_{x,y\in \D}|d(f(x),f(y))- \skp{x}{y}|\leq \del\Big)\geq 1-\eta.
	\end{equation}
	Then $m\gtrsim \del^{-2} \log(N/\eta)$.
\end{thm}
To prove this result, we will make use of ideas developed by Alon and Klartag in \cite{AlK17}. In that work they were concerned with the task of designing so-called \emph{$\del$-distance sketches}, that is, data structures which allow reconstruction of all Euclidean scalar products between points of a given $N$-point set $\mathcal{D}\subset B^n_2$ up to additive error $\del$. 
Using a probabilistic argument, they showed in \cite[Lemma~5.1]{AlK17} that if $n\sim \del^{-2}\log(N)$, then at least $c\del^{-2}N\log N$ bits are needed for any such data structure.  
Let us briefly sketch their argument, as it will form a key ingredient of our proof of Theorem~\ref{thm:optimalbit_inner_2}.  
First, they fix a maximal subset $\mathcal{N}\subset B^n_2$ consisting of $\tfrac{1}{2}$-separated points. It is well-known that 
$|\mathcal{N}|\sim 2^n$. In a second step, they choose independent random vectors $V_1, \ldots, V_{N/2}\in B^n_2$, which are uniformly distributed in $B^n_2$, and prove that with positive probability for any distinct points $x,x'\in \mathcal{N}$ there exists an $i\in [\tfrac{N}{2}]$ such that $|\skp{V_i}{x}-\skp{V_i}{x'}|> \del$. In particular, there exist at least $|\mathcal{N}|^{N/2}\sim 2^{nN/2}$ ordered $N$-point sets (each consisting of the vectors $V_1, \ldots, V_{N/2}$ and $\tfrac{N}{2}$-points from $\mathcal{N}$) whose Gram matrices differ pairwise in at least one entry by $\delta$. This immediatlely yields the lower bound on the number of bits needed for $\tfrac{\del}{2}$-distance sketches of $N$-point sets by using $n\sim \del^{-2}\log(N)$. 

In contrast to Alon and Klartag's setup, in Theorem~\ref{thm:optimalbit_inner_2} we are concerned with random encodings
$f$ (which can be seen to induce specific data structures) and aim to show a lower bound for the number of bits needed for the encoding of each individual vector, which additionally quantifies the dependence on the probability $\eta$ that the embedding fails. 
In order to do so, we will quantify and extend the argument by Alon and Klartag by showing the following:
\begin{enumerate}
\item For $k\leq n$, let $\mathcal{N}\subset \Sp^{k-1}$ be a maximal $\tfrac{1}{2}$-packing ($|\mathcal{N}|\sim 2^k$)
and for $M\sim \frac{N}{\eta}$ consider 
 independent vectors  
$V_1, \ldots ,V_{M}\in \Sp^{k-1}$ which are uniformly distributed in $\Sp^{k-1}$.
\item If $k\sim \del^{-2}\log(N/\eta)$ and the assumptions of Theorem~\ref{thm:optimalbit_inner_2} are satisfied then 
\begin{equation*}
\mathcal{B}=\big\{\forall x\neq  x'\in \mathcal{N} \; \exists i\in [M] : |\skp{x}{V_i}-\skp{x'}{V_i}|> 2\del\big\}\
\end{equation*}
occurs with probability at least $1-\exp(-32k)$.
\item If $X$ is uniformly distributed in $\mathcal{N}$ and \eqref{eq:optimalbit_inner_2:assumption} holds, then there exist 
	deterministic $\tilde{f}, \tilde{d}$ such that 
	\begin{equation*}
	\mathcal{A}= \Big\{\sup_{i\in [M]}|\tilde{d}(\tilde{f}(X),\tilde{f}(V_i))- \skp{X}{V_i}|\leq \del\Big\}
	\end{equation*}
	occurs with probability at least $c$, where $c$ is a suitable absolute constant.
\item On the event $\mathcal{B}$, 
	\begin{equation*}
	\Big|\Big\{x\in 
	\mathcal{N}\; : \; \sup_{i\in [M]}|\tilde{d}(\tilde{f}(x),\tilde{f}(V_i))- \skp{x}{V_i}|\leq \del\Big\} \Big|\leq |\{\tilde{f}(x)\; : \; x\in \mathcal{N}\}|\leq 2^m,
	\end{equation*}
	which implies $\bP(\mathcal{A}\cap \mathcal{B})\leq \tfrac{2^m}{|\mathcal{N}|}$. Part $(2)$ and $(3)$ imply that $\bP(\mathcal{A}\cap \mathcal{B})\geq \tfrac{c}{2}$ (for $k\gtrsim 1$), which shows the claim using $|\mathcal{N}|\sim 2^k\sim 2^{\del^{-2}\log(N/\eta)}$.	
\end{enumerate}

\begin{proof}[Proof of Theorem~\ref{thm:optimalbit_inner_2}]
	Set $s=\lceil 1/\eta\rceil\geq 2$.
	For $k\in [n]$, $k\gtrsim \log(N)$,
	let $\mathcal{N}\subset \Sp^{k-1}$ be a maximal $\tfrac{1}{2}$-packing of $\Sp^{k-1}$. Note that $\sqrt{k}\cdot2^{k-1}\leq |\mathcal{N}|\leq 9^k$ by Lemma~\ref{lem:packing:sphere}.
	Set $N'=N-1$, 
	let $V_1, \ldots, V_{sN'}$ denote independent random vectors which are uniformly distributed
	in $\Sp^{k-1}$ and set $S_j=\{V_{(j-1)N'+1}, \ldots, V_{j N'}\}$ for $j\in [s]$.
	Further, define the independent events 
	\begin{equation*}
	\mathcal{E}_{S_j}:=\bigg\{\inf_{x\neq y\in S_j\cup \mathcal{N}} \eu{x-y} >\alpha\bigg\}.
	\end{equation*}
	Since $k\gtrsim \log(N)$ and $|\mathcal{N}|\geq 2^{k-1}$ it follows $|\mathcal{N}|\geq N'=|S_j|$.
	Using inequality \eqref{eq:sphere:small_ball} we obtain
	\begin{align*}
	\bP(\mathcal{E}_{S_j}^c)&\leq 
	\sum_{x\in \mathcal{N}}\sum_{y\in S_j}\bP(\eu{x-y}\leq \alpha) + \sum_{x\neq y\in S_j}\bP(\eu{x-y}\leq \alpha)\\
	&\leq 2|\mathcal{N}|^2\alpha^{k-1} \leq 2^{-k},
	\end{align*}
	where the last inequality follows by choosing $\alpha>0$ small enough.
	Define the event 
	\begin{equation*}
	\mathcal{B}=\big\{\forall x\neq  x'\in \mathcal{N} \; \exists i\in [sN'] : |\skp{x}{V_i}-\skp{x'}{V_i}|> 2\del\big\}.
	\end{equation*}
	Let us show that if 
	\begin{equation}\label{eq:cond:k} 
	\tfrac{1}{8}\del^{-2}\leq k\leq \tfrac{1}{64} \del^{-2}\log(sN'),\quad \del^{-2}\log(sN')\leq \sqrt{sN'},
	\end{equation}
	then 
	\begin{equation*}
	\Pb(\mathcal{B})\geq 1-\exp(-32k).
	\end{equation*}
	Using the independence of the vectors $V_i$ and the union bound we obtain
	\begin{align*}
	\bP(\mathcal{B}^c)=&\Pb(\exists x\neq  x'\in \mathcal{N} \; : \; \forall i\in [sN'] : |\skp{x}{V_i}-\skp{x'}{V_i}|\leq 2\del)\\
	&\leq \sum_{x\neq  x'\in \mathcal{N}} \big(\Pb(|\skp{x}{V}-\skp{x'}{V}|\leq 2\del)\big)^{sN'}
	\end{align*}
	where $V$ is uniformly distributed in $\Sp^{k-1}$. Since 
	$\eu{x-x'}> \frac{1}{2}$ for any $x\neq x'\in \mathcal{N}$,
	\begin{equation*}
	\Pb(|\skp{x}{V}-\skp{x'}{V}|> 2\del)
	\geq \Pb\left(\left|\left\langle\frac{x-x'}{\eu{x-x'}},V\right\rangle\right|> 4\del\right)\geq \exp(-32\del^2k)
	\end{equation*}
	where in the last step we applied Lemma~\ref{lem:sphere:uniformmeasure}.
	Hence, 
	\begin{align*}
	\Pb(\mathcal{B}^c)
	&\leq |\mathcal{N}|^2 \big(1-\exp(-32\del^2k)\big)^{sN'}\\
	&\leq \exp(\log(81)k)\cdot\big(\exp(-\exp(-32\del^2k))\big)^{sN'}\\
	&=\exp(\log(81)k - sN'\exp(-32\del^2k)).
	\end{align*}
	Finally, using \eqref{eq:cond:k} we obtain
	\begin{equation*}
	\log(81)k - sN'\exp(-32\del^2k)\leq -32k.
	\end{equation*}
In the following, we consider $k=\tfrac{1}{64}\del^{-2}\log(sN')$ and assume $\del^{-2}\log(sN')\leq (sN')^{1/2}$.
Let $X$ be uniformly distributed in $\mathcal{N}$. Let $S=\{V_1,\ldots, V_{N'}\}$ be a set 
of independent random vectors $V_i$ which are uniformly distributed on $\Sp^{k-1}$.
By independence of the random vectors $V_1, \ldots, V_{sN'}$ and Jensen's inequality, 
	\begin{align*}
	&\bP\Big( \cap_{j\in [s]}\mathcal{E}_{S_j}, \sup_{y\in \{V_{1}, \ldots, V_{sN'}\}}|d(f(X),f(y))- \skp{X}{y}|\leq \del\Big)\\
	&=\E_{f,d,X}\bP_V\Big( \cap_{j\in [s]} \big\{\sup_{y\in S_j}|d(f(X),f(y))- \skp{X}{y}|\leq \del, \, \mathcal{E}_{S_j}\big\}\Big)\\
	&=\E_{f,d,X} \Big[\bP_S\Big( \sup_{y\in S}|d(f(X),f(y))- \skp{X}{y}|\leq \del, \, \mathcal{E}_S\Big)\Big]^s\\
	&\geq \Big[\bP\Big( \sup_{y\in S}|d(f(X),f(y))- \skp{X}{y}|\leq \del, \, \mathcal{E_S}\Big)\Big]^s.
	\end{align*}
	Observe that $S\cup \{X\}\subset B^n_2$ satisfies 
	$|S\cup \{X\}|\leq N$. Therefore, \eqref{eq:optimalbit_inner_2:assumption} implies
	\begin{align*}
	&\bP\Big( \sup_{y\in S}|d(f(X),f(y))- \skp{X}{y}|\leq \del, \, \mathcal{E_S}\Big)
	\geq (1-\eta) \bP(\mathcal{E_S})
	\geq (1-\eta) (1-2^{-k}).
	\end{align*}
	Hence,
	\begin{align*}
	\bP\Big( \cap_{j\in [s]}\mathcal{E}_j, \sup_{y\in \{V_{1}, \ldots, V_{sN'}\}}|d(f(x),f(y))- \skp{x}{y}|\leq \del\Big)\geq (1-\eta)^s (1-2^{-k})^s.
	\end{align*}
	Observe that $1-\eta\geq \exp(-2\eta)$ for $\eta\in [0,\frac{1}{2}]$. Hence, if $\eta\in (0, \tfrac{1}{2}]$, then $(1-\eta)^s\geq \exp(-2\eta s)\geq \exp(-4)$. Further, if $k\geq \log(\eta^{-1})$, then $(1-2^{-k})^s\geq \frac{1}{2}$.
	Hence, 
	\begin{equation*}
	e^{-4}\cdot2^{-1}\leq \bP\Big( \sup_{y\in \{V_{1}, \ldots, V_{sN'}\}}|d(f(X),f(y))- \skp{X}{y}|\leq \del\Big).
	\end{equation*}
	This implies that there exist (deterministic) $\tilde{f}: B^n_2\to \{-1,1\}^m$ and \\$\tilde{d}:\{-1,1\}^m\times \{-1,1\}^m\to \R$ such that 
	\begin{equation*}
	e^{-4}\cdot2^{-1}\leq \bP_{X,V}\Big( \sup_{y\in \{V_{1}, \ldots, V_{sN'}\}}|\tilde{d}(\tilde{f}(X),\tilde{f}(y))- \skp{X}{y}|\leq \del\Big).
	\end{equation*}
	Since $X$ is uniformly distributed in $\mathcal{N}$,
	\begin{align*}
	& e^{-4}\cdot2^{-1} - \bP(\mathcal{B}^c) \\
	&\qquad \leq \Pb_{X,V}\Big( \sup_{y\in \{V_{1}, \ldots, V_{sN'}\}}|\tilde{d}(\tilde{f}(X),\tilde{f}(y))- \skp{X}{y}|\leq \del,\; \mathcal{B}\Big)\\
	&\qquad =  \E_V \Big[1_{\mathcal{B}}\cdot \Pb_{X}\Big( \sup_{y\in \{V_{1}, \ldots, V_{sN'}\}}|\tilde{d}(\tilde{f}(X),\tilde{f}(y))- \skp{X}{y}|\leq \del \Big)\Big]\\
	&\qquad =\frac{1}{|\mathcal{N}|} \E_V \Big[1_{\mathcal{B}} \cdot \Big|\Big\{x\in 
	\mathcal{N}\; : \; \sup_{y\in \{V_{1}, \ldots, V_{sN'}\}}|\tilde{d}(\tilde{f}(x),\tilde{f}(y))- \skp{x}{y}|\leq \del\Big\} \Big|\Big].
	\end{align*}
	Let us now show that on the event $\mathcal{B}$
	\begin{equation*}
	\Big|\Big\{x\in 
	\mathcal{N}\; : \; \sup_{y\in \{V_{1}, \ldots, V_{sN'}\}}|\tilde{d}(\tilde{f}(x),\tilde{f}(y))- \skp{x}{y}|\leq \del\Big\} \Big|\leq |\{\tilde{f}(x)\; : \; x\in \mathcal{N}\}|.
	\end{equation*}
	It suffices to show that for any $b\in \{\tilde{f}(x)\; : \; x\in \mathcal{N}\}$, 
	\begin{equation}\label{eq:proof:thm:optimalbit_inner_2:sum}
	\sum_{x\in \tilde{f}^{-1}(\{b\})\cap \mathcal{N}}1_{\big\{\sup_{y\in \{V_{1}, \ldots, V_{sN'}\}}|\tilde{d}(\tilde{f}(x),\tilde{f}(y))- \skp{x}{y}|\leq \del \big\}}\leq 1.
	\end{equation}
	To prove \eqref{eq:proof:thm:optimalbit_inner_2:sum} assume that $\sup_{y\in \{V_{1}, \ldots, V_{sN'}\}}|\tilde{d}(b,\tilde{f}(y))- \skp{x}{y}|\leq \del$ for some $x\in \tilde{f}^{-1}(\{b\})\cap \mathcal{N}$. By definition of the event $\mathcal{B}$, for any $x'\in \tilde{f}^{-1}(\{b\})\cap \mathcal{N}$ with $x'\neq x$ there exists a $y\in \{V_{1}, \ldots, V_{sN'}\}$ such that $|\skp{x}{y}-\skp{x'}{y}|> 2\del$. This implies 
	\begin{align*}
	2\del < |\skp{x}{y}-\skp{x'}{y}|&\leq |\skp{x}{y}- \tilde{d}(b,\tilde{f}(y))| + |\tilde{d}(b,\tilde{f}(y))-\skp{x'}{y}|\\
	&\leq \del + |\tilde{d}(b,\tilde{f}(y))-\skp{x'}{y}|
	\end{align*}
	and therefore
	\begin{equation*}
	\del< |\tilde{d}(\tilde{f}(x'),\tilde{f}(y))-\skp{x'}{y}|.
	\end{equation*}
	Hence, 
	\begin{equation*}
	1_{\big\{\sup_{y\in \{V_{1}, \ldots, V_{sN'}\}}|\tilde{d}(\tilde{f}(x'),\tilde{f}(y))- \skp{x'}{y}|\leq \del\big\}}=0,
	\end{equation*}
	which proves \eqref{eq:proof:thm:optimalbit_inner_2:sum}.
	In summary, we have shown that if $k=\tfrac{1}{64}\del^{-2}\log(sN')$ and $\del^{-2}\log(sN')\leq \sqrt{sN'}$, then 
	\begin{equation*}
	\frac{|\mathcal{N}|}{4e^4}\leq\big(e^{-4}\cdot2^{-1}-e^{-32k}\big)|\mathcal{N}|\leq\big(e^{-4}\cdot2^{-1}- \bP(\mathcal{B}^c)\big) |\mathcal{N}| \leq |\{\tilde{f}(x)\; : \; x\in \mathcal{N}\}|\leq 2^m.
	\end{equation*}
	Since $|\mathcal{N}|\geq 2^k$, this yields the result.
\end{proof}

\begin{thm}\label{thm:lowerbound:continuous}
	There exists an absolute constant $\alpha\in (0,\tfrac{1}{2})$ such that the following holds. 
	Let $\del\in (0, \alpha), \eta\in (0,\tfrac{1}{2}), N_\del\in \N, w_\del>0$ and assume that 
	\begin{equation}
	n\gtrsim \del^{-2} \log(N_\del/\eta)+ \del^{-2} w_\del^2, \qquad \sqrt{N_\del/\eta}\gtrsim \del^{-2}\log(N_\del/\eta), \qquad \del^{-2}w_\del^2\gtrsim 1.
	\end{equation}
	Let $f: B^n_2\to \{-1,1\}^m$ be a random binary embedding map and $d:\{-1,1\}^m\times \{-1,1\}^m\to \R$ a random reconstruction map such that given any $\D\subset B^n_2$ 
	with 
	\begin{equation*}
	\mathcal{N}(\D, 8\del)\leq N_\del\; \text{ and }\;  \ell_*((\D-\D)\cap 8\del B^n_2)\leq w_\del
	\end{equation*}
	we have
	\begin{equation}\label{eq:thm:lowerbound:continuous}
	\bP\Big(\sup_{x,y\in \D}|d(f(x),f(y))- \skp{x}{y}|\leq \del\Big)\geq 1-\eta.
	\end{equation}
	Further, assume that the quantization cells of $f$ are convex.
	Then $m\gtrsim \del^{-2} \log(N_\del/\eta) + \del^{-2}w_\del^2$.
\end{thm}

\begin{proof} 
	Let $\alpha\in (0,\tfrac{1}{2})$ denote the absolute constant from Theorem~\ref{thm:optimalbit_inner_2}.
	If $8\del\leq \alpha$ then every $\D\subset B^n_2$ with $|\D|=N_\del$ and $(\D-\D)\cap \alpha B^n_2=\{0\}$
	satisfies
	\begin{equation*}
	\mathcal{N}(\D, 8\del)\leq N_\del \text{ and }  \ell_*((\D-\D)\cap 8\del B^n_2)=0\leq w_\del,
	\end{equation*}
	so by assumption
	\begin{equation*}
	\bP\Big(\sup_{x,y\in \D}|d(f(x),f(y))- \skp{x}{y}|\leq \del\Big)\geq 1-\eta.
	\end{equation*}
	By Theorem~\ref{thm:optimalbit_inner_2} this implies $m\gtrsim \del^{-2} \log(N_\del/\eta)$. 
	Next, let us show that if $n\geq \del^{-2}w_\del^2$
	and 
	\begin{equation*}
	\bP\Big(\sup_{x,y\in \D}|d(f(x),f(y))- \skp{x}{y}|\leq \del\Big)\geq 1-\eta
	\end{equation*}
	for every $\D\subset B^n_2$ with $\mathcal{N}(\D, 8\del)=1$ and $\ell_*((\D-\D)\cap 8\del B^n_2)\leq w_\del$, then 
	\begin{equation*}
	m\geq \del^{-2}w_\del^2 - 1.
	\end{equation*}
	Set $k= c\del^{-2}w_{\del}^2$ for a constant $c>0$
	and let $E\subset \R^n$ denote a $k$-dimensional subspace.
	Let $\mathcal{N}\subset E\cap B^n_2$ be a maximal $\tfrac{1}{2}$-packing of $E\cap B^n_2$, then $2^k\leq |\mathcal{N}|\leq 9^k$ (e.g. see \cite[Lemma 4.2.8 and Corollary 4.2.13]{RomanHDP}).
	For $x\in \R^n$ define $\D_x= x + (E\cap 8\del  B^n_2)$. 
	Then $\mathcal{N}(\D_x, 8\del)=1$ and $\ell_*((\D_x-\D_x)\cap 8\del B^n_2)\leq 2\ell_*(E\cap 8\del  B^n_2)\lesssim \del \sqrt{k}$.
	Hence, if $c>0$ is chosen small enough, then $\ell_*((\D_x-\D_x)\cap 8\del B^n_2)\leq w_\del$. Let $X$ be uniformly distributed in $\mathcal{N}$.
	By assumption,
	\begin{equation*}
	1-\eta\leq \bP_{X} \bP_{f,d}\Big(\sup_{z, z'\in \D_{X}} |d(f(z),f(z'))- \skp{z}{z'}|\leq \del\Big).
	\end{equation*}
	Hence, there exist (deterministic) $\tilde{f}: B^n_2\to \{-1,1\}^m$ and \\$\tilde{d}:\{-1,1\}^m\times \{-1,1\}^m\to \R$ such that 
	\begin{equation*}
	1-\eta\leq \bP_{X}\Big(\sup_{z, z'\in \D_{X}}|\tilde{d}(\tilde{f}(z),\tilde{f}(z'))- \skp{z}{z'}|\leq \del\Big).
	\end{equation*}
	Clearly, 
	\begin{align*}
	&\bP_{X}\Big(\sup_{z, z'\in \D_{X}}|\tilde{d}(\tilde{f}(z),\tilde{f}(z'))- \skp{z}{z'}|\leq \del\Big)\\
	&\qquad =\frac{1}{|\mathcal{N}|}\sum_{x\in \mathcal{N}}1_{\big\{  \sup_{v, v'\in E\cap 8\del  B^n_2}|\tilde{d}(\tilde{f}(x+v),\tilde{f}(x+v'))- \skp{x+v}{x+v'}|\leq \del \big\}}\\
	&\qquad \leq\frac{1}{|\mathcal{N}|}\sum_{x\in \mathcal{N}}1_{\big\{  \sup_{v\in E\cap 8\del  B^n_2}|\tilde{d}(\tilde{f}(x),\tilde{f}(x+v))- \skp{x}{x+v}|\leq \del \big\}}.
	\end{align*}
	Next, let us show
	\begin{equation*}
	\sum_{x\in \mathcal{N}}1_{\big\{  \sup_{v\in E\cap 8\del  B^n_2}|\tilde{d}(\tilde{f}(x),\tilde{f}(x+v))- \skp{x}{x+v}|\leq \del \big\}}\leq |\{\tilde{f}(x)\; : \; x\in \mathcal{N}\}|.
	\end{equation*}
	Clearly it suffices to prove that 
	\begin{align}\label{eq:proof:thm:lowerbound:continuous}
	&1_{\big\{\sup_{v\in E\cap 8\del  B^n_2}|\tilde{d}(\tilde{f}(x),\tilde{f}(x+v))- \skp{x}{x+v}|\leq \del\big\}}\\
	&\quad +1_{\big\{\sup_{v\in E\cap 8\del  B^n_2}|\tilde{d}(\tilde{f}(x'),\tilde{f}(x'+v))- \skp{x'}{x'+v}|\leq \del\big\}}\leq 1 \nonumber
	\end{align}
	for all $x\neq x'\in \mathcal{N}$ with $\tilde{f}(x)=\tilde{f}(x')$. Without loss of generality, we assume $\eu{x}\leq \eu{x'}$ and set
	$\tilde{v}=8\del\frac{x'-x}{\eu{x'-x}}\in E\cap 8\del  B^n_2$. If $8\del \leq \frac{1}{2}$,
	then $\tfrac{8\del}{\eu{x-x'}}\leq 1$ which implies $x+\tilde{v}\in \operatorname{Span}(\{x,x'\})$. 
	Since the quantization cells of $\tilde{f}$ are convex it follows $\tilde{f}(x+\tilde{v})=\tilde{f}(x)=\tilde{f}(x')$.
	Further, using $\eu{x}\leq \eu{x'}$ we find
	\begin{align*}
	\skp{x'}{x'} - \skp{x}{x+\tilde{v}} &= \skp{x'}{x'} - \skp{x}{x} + \frac{8\del}{\eu{x-x'}}\skp{x}{x-x'}\\
	&= \frac{8\del}{\eu{x-x'}}(\skp{x'}{x'} - \skp{x}{x}+ \skp{x}{x-x'})\\
	&\qquad \quad + \Big[1-\frac{8\del}{\eu{x-x'}}\Big](\skp{x'}{x'} - \skp{x}{x})\\
	&\geq \frac{8\del}{2\eu{x-x'}}(2\skp{x'}{x'} - 2\skp{x}{x'})\\
	&\geq 4\del \eu{x-x'}\\
	&> 2\del.
	\end{align*}
	Therefore, $1_{\{|\tilde{d}(\tilde{f}(x),\tilde{f}(x+\tilde{v}))- \skp{x}{x+\tilde{v}}|\leq \del\}} + 1_{\{|\tilde{d}(\tilde{f}(x'),\tilde{f}(x'))- \skp{x'}{x'}|\leq \del\}}\leq 1$,
	which implies \eqref{eq:proof:thm:lowerbound:continuous}.
	In summary, we showed $(1-\eta)|\mathcal{N}|\leq |\{\tilde{f}(x)\; : \; x\in \mathcal{N}\}|\leq 2^m$. Since $\eta<\tfrac{1}{2}$ and $2^k\leq |\mathcal{N}|$, this yields $m\geq c\del^{-2}w_\del^2 - 1\geq \tfrac{c}{2}\del^{-2}w_\del^2$, provided that $\tfrac{c}{2}\del^{-2}w_\del^2\geq 1$.
\end{proof}

\section*{Acknowledgements}

S.D.\ was supported by the Deutsche Forschungsgemeinschaft (DFG, German Research Foundation) under SPP 1798 (COSIP - Compressed Sensing in Information Processing) through project CoCoMiMo.  A.S.\ acknowledges support by the Fonds de la Recherche Scientifique - FNRS under Grant n$^{\circ}$ T.0136.20 (Learn2Sense).

\end{document}